\documentclass[12pt, onecolumn,a4paper,accepted=2022-09-31]{quantumarticle}
\pdfoutput=1 
\usepackage[T1]{fontenc}

\usepackage{tikz}
\usetikzlibrary{positioning, calc}
\usetikzlibrary{calc}
\usetikzlibrary{arrows}
\usepackage{tikz-3dplot}
\usepackage{graphicx}
\usepackage{mdwlist}
\usepackage{color}
\usepackage{thmtools}
\usepackage{amssymb}
\usepackage{physics}
\usepackage{amsmath}
\usepackage{array}
\usepackage{doi}
\usepackage{url,hyperref}
\usepackage[numbers,sort&compress]{natbib}
\usetikzlibrary{fadings}
\usetikzlibrary{decorations.pathmorphing}
\usepackage{mathrsfs}
\usepackage{authblk}
\usepackage{amsmath}
\usepackage[mathscr]{euscript}
\usepackage{enumitem}

\DeclareMathAlphabet\mathbfcal{OMS}{cmsy}{b}{n}

\usepackage{cleveref}

\usetikzlibrary{decorations.pathreplacing}
\tikzset{snake it/.style={decorate, decoration=snake}}

\usetikzlibrary{decorations.pathreplacing,decorations.markings}

\tikzset{
    >=stealth',
    punkt/.style={
           rectangle,
           rounded corners,
           draw=black, very thick,
           text width=6.5em,
           minimum height=2em,
           text centered},
    pil/.style={
           ->,
           thick,
           shorten <=2pt,
           shorten >=2pt,},
  on each segment/.style={
    decorate,
    decoration={
      show path construction,
      moveto code={},
      lineto code={
        \path [#1]
        (\tikzinputsegmentfirst) -- (\tikzinputsegmentlast);
      },
      curveto code={
        \path [#1] (\tikzinputsegmentfirst)
        .. controls
        (\tikzinputsegmentsupporta) and (\tikzinputsegmentsupportb)
        ..
        (\tikzinputsegmentlast);
      },
      closepath code={
        \path [#1]
        (\tikzinputsegmentfirst) -- (\tikzinputsegmentlast);
      },
    },
  },
  mid arrow/.style={postaction={decorate,decoration={
        markings,
        mark=at position .5 with {\arrow[#1]{stealth'}}
      }}}
}

\usepackage{hyperref}
\usepackage{slashed}

\usepackage{float}		
\usepackage{graphicx}		
\usepackage{subcaption}
\usepackage{amsmath}

\DeclareMathOperator{\Mod}{mod}

\newcommand{\defi}{:=}
\usepackage{bbold}

\newcommand{\moddl}{\text{Mod}_d\text{L}}

\newcommand{\spacet}{\text{SPACE}_{(2)}}

\newtheorem{theorem}{Theorem}

\newtheorem{conjecture}[theorem]{Conjecture}

\newtheorem{definition}[theorem]{Definition}

\newtheorem{lemma}[theorem]{Lemma}

\newtheorem{remark}[theorem]{Remark}

\newtheorem{protocol}[theorem]{Protocol}
\newenvironment{proof}[1][Proof]{\noindent\textbf{#1.}}{\ \rule{0.5em}{0.5em}}

\begin{document} 

\title{Complexity and entanglement in non-local computation and holography}

\author[1]{Alex May}
\orcid{0000-0002-4030-5410}

\affil[1]{Stanford University}

\begin{abstract}
Does gravity constrain computation? 
We study this question using the AdS/CFT correspondence, where computation in the presence of gravity can be related to non-gravitational physics in the boundary theory.
In AdS/CFT, computations which happen locally in the bulk are implemented in a particular non-local form in the boundary, which in general requires distributed entanglement. 
In more detail, we recall that for a large class of bulk subregions the area of a surface called the ridge is equal to the mutual information available in the boundary to perform the computation non-locally. 
We then argue the complexity of the local operation controls the amount of entanglement needed to implement it non-locally, and in particular complexity and entanglement cost are related by a polynomial. 
If this relationship holds, gravity constrains the complexity of operations within these regions to be polynomial in the area of the ridge. 
\end{abstract}

\vfill

\maketitle

\pagebreak

\tableofcontents

\section{Introduction}

Quantum information theory addresses the fundamental limits and capabilities of information processing given the rules of quantum mechanics. 
However, to fully understand information processing we need to go beyond quantum mechanics: our universe is described more completely by a theory of quantum gravity.
How can we understand the theory of quantum-gravitational information processing?

Progress in understanding information processing in the presence of gravity has been made on a number of fronts. 
Most famously, a basic question is whether unitarity is preserved in the context of gravity, and recent progress has supported that it is \cite{penington2020entanglement,almheiri2019entropy}. 
Others have asked about the computational properties of quantum field theories, with an eventual goal of understanding the computational power of quantum gravity \cite{jordan2018bqp}. 
As we discuss in more detail later on, a well-known claim is that gravity should constrain computation \cite{lloyd2000ultimate}. However, existing arguments for this idea have flaws \cite{jordan2017arbitrarily}.

A precise framework within which we can address questions about quantum gravity is the AdS/CFT correspondence \cite{maldacena1999large,witten1998anti}. 
There, a theory of quantum gravity in asymptotically $d+1$ dimensional anti de Sitter space (``the bulk'') is given an alternative, non-gravitational description as a conformal field theory, which lives in the boundary of that spacetime. 
See e.g. \cite{harlow2018tasi} for a review. 
In this setting, we can ask questions about information processing in the bulk, gravitating, theory by translating them into non-gravitational questions in the boundary. 

In the AdS/CFT correspondence quantum-gravitational physics is captured by a purely quantum mechanical description in the boundary. 
Because of this, we could expect that quantum-gravitational information processing is captured by standard quantum information theory --- we need only note the conformal field theory is quantum mechanical, so described by the usual theory of quantum information, and we are apparently done. 
However, the true situation is more subtle. 
The gravitational and non-gravitational descriptions have additional structure beyond quantum mechanics alone: they have causal structure, and this causal structure constrains information processing. 
Further, the same information processing task viewed in the bulk and boundary perspectives will have different causal structures. 
We argue that consequently gravity implies constraints on information processing, beyond those that a naive application of quantum mechanics would suggest. 

\begin{figure*}
\begin{center}
\begin{subfigure}[b]{.45\textwidth}
\begin{center}
\begin{tikzpicture}[scale=0.55]
    
    \draw[postaction={on each segment={mid arrow}}] (-4,0) -- (0,4);
    \draw[postaction={on each segment={mid arrow}}] (4,0) -- (0,4);
    \draw[postaction={on each segment={mid arrow}}] (0,4) -- (4,8);
    \draw[postaction={on each segment={mid arrow}}] (0,4) -- (-4,8);
    
    \node[below left] at (-4,0) {$\mathcal{C}_0$};
    \draw[fill=black] (-4,0) circle (0.15);

    \node[below right] at (4,0) {$\mathcal{C}_1$};
    \draw[fill=black] (4,0) circle (0.15);

    \node[above right] at (4,8) {$\mathcal{R}_1$};
    \draw[fill=blue] (4,8) circle (0.15);

    \node[above left] at (-4,8) {$\mathcal{R}_0$};
    \draw[fill=blue] (-4,8) circle (0.15);
    
    \draw[fill=yellow] (0,4) circle (0.30);
    
    \node[below] at (0,-0.53) {$ $};
    
    \node[above left] at (-3,1) {$A_0$};
    \node[above right] at (3,1) {$A_1$};
    
    \node[below left] at (-3,7) {$B_0$};
    \node[below right] at (3,7) {$B_1$};
    
    \end{tikzpicture}
\end{center}
\caption{}
\label{subfig:localschematic}
\end{subfigure}
\hfill
\begin{subfigure}[b]{.45\textwidth}
\begin{center}
\begin{tikzpicture}[scale=0.55]

    \node[below left] at (-4,0) {$\mathcal{C}_0$};
    \draw[fill=black] (-4,0) circle (0.15);

    \node[below right] at (4,0) {$\mathcal{C}_1$};
    \draw[fill=black] (4,0) circle (0.15);

    \node[above right] at (4,8) {$\mathcal{R}_1$};
    \draw[fill=blue] (4,8) circle (0.15);

    \node[above left] at (-4,8) {$\mathcal{R}_0$};
    \draw[fill=blue] (-4,8) circle (0.15);
    
    \draw[postaction={on each segment={mid arrow}}] (-4,0) -- (-2,2) -- (-2,6) -- (-4,8);
    \draw[postaction={on each segment={mid arrow}}] (4,0) -- (2,2) -- (2,6) -- (4,8);
    \draw[postaction={on each segment={mid arrow}}] (-2,2) -- (0,4) -- (2,6);
    \draw[postaction={on each segment={mid arrow}}] (2,2) -- (0,4) -- (-2,6);
    
    \draw[dashed] (2,2) -- (0,0) -- (-2,2);
    \node[below] at (0,0) {$\Psi_{LR}$};
    
    \draw[fill=yellow] (-2,2) circle (0.3);
    \draw[fill=yellow] (2,2) circle (0.3);
    \draw[fill=yellow] (-2,6) circle (0.3);
    \draw[fill=yellow] (2,6) circle (0.3);
    
    \node[above left] at (-3,1) {$A_0$};
    \node[above right] at (3,1) {$A_1$};
    
    \node[below left] at (-3,7) {$B_0$};
    \node[below right] at (3,7) {$B_1$};
    
\end{tikzpicture}
\end{center}
\caption{}
\label{subfig:nonlocalschematic}
\end{subfigure}
\begin{subfigure}[b]{0.45\textwidth}
\centering
\tdplotsetmaincoords{15}{0}
\begin{tikzpicture}[scale=1.05,tdplot_main_coords]
\tdplotsetrotatedcoords{0}{35}{0}
\draw[gray] (-2,0.5,0) -- (-2,6.25,0);
\draw[gray] (2,0.5,0) -- (2,6.25,0);
    
    \begin{scope}[tdplot_rotated_coords]
    
    \draw[domain=0:45,variable=\x,smooth, fill=black!60!,opacity=0.8] plot ({-2*sin(\x)}, {1+\x/45}, {2*cos(\x)}) -- plot ({-2*sin((45-\x))}, {3-(45-\x)/45}, {2*cos(45-\x)}) --  plot ({2*sin(\x)}, {3-\x/45}, {2*cos(\x)}) -- plot ({2*sin(45-\x)}, {1+(45-\x)/45}, {2*cos(45-\x)});
    
    \draw[domain=0:45,variable=\x,smooth,thick] plot ({-2*sin(\x)}, {1+\x/45}, {2*cos(\x)});
    \draw[domain=0:45,variable=\x,smooth,thick] plot ({2*sin(\x)}, {1+\x/45}, {2*cos(\x)});
    \draw[domain=0:45,variable=\x,smooth,thick] plot ({-2*sin(\x)}, {3-\x/45}, {2*cos(\x)});
    \draw[domain=0:45,variable=\x,smooth,thick] plot ({2*sin(\x)}, {3-\x/45}, {2*cos(\x)});
    
    \begin{scope}[canvas is xz plane at y=0.5]
    \draw[gray] (0,0) circle[radius=2] ;
    \end{scope}
    
    \begin{scope}[canvas is xz plane at y=6.25]
    \draw[gray] (0,0) circle[radius=2] ;
    \end{scope}
    
    \draw[red] (0,1,-2) -- (0,2,-1);
    \draw[red] (0,1,2) -- (0,2,1);
    
    \begin{scope}[canvas is xz plane at y=2]
    
    \draw[gray] (0,0) circle (2);
    
    
    \draw [green,ultra thick,domain=-45:45] plot ({2*cos(\x+90)}, {2*sin(\x+90)});
    
    
    \end{scope}
    
    \draw[domain=0:90,variable=\x,smooth,dashed] plot ({2*sin(\x+180)}, {3+\x/45}, {2*cos(\x+180)});
    \draw[domain=0:90,variable=\x,smooth,dashed] plot ({2*sin(\x+180)}, {3+\x/45}, {-2*cos(\x+180)});
    \draw[domain=0:90,variable=\x,smooth,dashed] plot ({-2*sin(\x+180)}, {3+\x/45}, {2*cos(\x+180)});
    \draw[domain=0:90,variable=\x,smooth,dashed] plot ({-2*sin(\x+180)}, {3+\x/45}, {-2*cos(\x+180)});
    
    \draw[thick, red,-triangle 45] (0,3,0) -- (2,5,0);
    \draw[thick, red,-triangle 45] (0,3,0) -- (-2,5,0);
    
    \draw[thick,red,-triangle 45] (0,2,-1) -- (0,3,0);
    \draw[thick,red,-triangle 45] (0,2,1) -- (0,3,0);
    
    \draw plot [mark=*, mark size=1.5] coordinates{(2,5,0)};
    \node[above] at (2,5,0) {$\mathcal{R}_0$};
    \draw plot [mark=*, mark size=1.5] coordinates{(-2,5,0)};
    \node[above] at (-2,5,0) {$\mathcal{R}_1$};
    
    \draw[domain=0:45,variable=\x,smooth, fill=black!50!,opacity=0.8] plot ({-2*sin(\x+180)}, {1+\x/45}, {2*cos(\x+180)}) -- plot ({-2*sin((45-\x)+180)}, {3-(45-\x)/45}, {2*cos(45-\x+180)}) --  plot ({2*sin(\x+180)}, {3-\x/45}, {2*cos(\x+180)}) -- plot ({2*sin(45-\x+180)}, {1+(45-\x)/45}, {2*cos(45-\x+180)});
    
    \begin{scope}[canvas is xz plane at y=2]
    \draw [green,ultra thick,domain=-45:45] plot ({2*cos(\x-90)}, {2*sin(\x-90)});
    \end{scope}
    
    \draw[domain=0:45,variable=\x,smooth,thick] plot ({-2*sin(\x+180)}, {1+\x/45}, {2*cos(\x+180)});
    \draw[domain=0:45,variable=\x,smooth,thick] plot ({2*sin(\x+180)}, {1+\x/45}, {2*cos(\x+180)});
    \draw[domain=0:45,variable=\x,smooth,thick] plot ({-2*sin(\x+180)}, {3-\x/45}, {2*cos(\x+180)});
    \draw[domain=0:45,variable=\x,smooth,thick] plot ({2*sin(\x+180)}, {3-\x/45}, {2*cos(\x+180)});
    
    \draw plot [mark=*, mark size=1.5] coordinates{(0,1,-2)};
    \node[left] at (0,1,-2) {$\mathcal{C}_0$};
    \draw plot [mark=*, mark size=1.5] coordinates{(0,1,2)};
    \node[below] at (0,1,2) {$\mathcal{C}_1$};
    \draw plot [mark=*, mark size=1.5] coordinates{(0,3,0)};
    \node[left] at (0,3,0) {$J_{01\rightarrow 01}$};
    
    \end{scope}
    
    \begin{scope}[canvas is xz plane at y=2]
    \draw[lightgray,domain=10:45] plot ({2*cos(\x-90)}, {2*sin(\x-90)});;
    \end{scope}
    
    \end{tikzpicture}
    
\caption{}
\label{subfig:cylinder}
\end{subfigure}
\hfill
\begin{subfigure}[b]{.45\textwidth}
\begin{center}
\begin{tikzpicture}[scale=1.5]

 \draw (-2,0) -- (2,0) -- (2,2) -- (-2,2) -- (-2,0);
    
    \draw[gray,fill=gray,opacity=0.25] (-2,2) -- (0,0) -- (2,2);
    
    \draw[gray,fill=gray,opacity=0.25] (-2,2) -- (-2,0) -- (0,2) -- (-2,2);
    
    \draw[gray,fill=gray,opacity=0.25] (2,2) -- (2,0) -- (0,2) -- (2,2);
    
    \draw[blue,fill=blue,opacity=0.25] (-1,2) -- (-2,1) -- (-2,0) -- (1,0) -- (-1,2);
    
    \draw[blue,fill=blue,opacity=0.25] (1,2) -- (2,1) -- (2,0) -- (-1,0) -- (1,2);
    
    \draw[blue,fill=blue,opacity=0.25] (2,1) -- (2,0) -- (1,0) -- (2,1);
    \draw[blue,fill=blue,opacity=0.25] (-2,1) -- (-2,0) -- (-1,0) -- (-2,1);
    
    \draw[black] plot [mark=*, mark size=2] coordinates{(-2,0)};
    \node[below left] at (-2,0) {$\mathcal{C}_1$};
    
    \draw[black] plot [mark=*, mark size=2] coordinates{(2,0)};
    \node[below right] at (2,0) {$c_1$};
    
    \draw[black] plot [mark=*, mark size=2] coordinates{(0,0)};
    \node[below left] at (0,0) {$\mathcal{C}_0$};
    
    \draw[blue] plot [mark=*, mark size=2] coordinates{(-1,2)};
    \node[above right] at (-1,2) {$\mathcal{R}_0$};
    
    \draw[blue] plot [mark=*, mark size=2] coordinates{(1,2)};
    \node[above right] at (1,2) {$\mathcal{R}_1$};
    
    \node at (0,-1) {$ $};

\end{tikzpicture}
\end{center}
\caption{}
\label{subfig:boundary}
\end{subfigure}
\caption{(a) A computation on systems $A_0$, $A_1$ happening locally. The yellow circle represents a channel acting on input systems $A_0$ and $A_1$, and producing output systems $B_0$ and $B_1$. (b) A computation happening non-locally. $A_0$ is interacted with the $L$ system, and $A_1$ with the $R$ system, where $\Psi_{LR}$ is entangled. Then, a round of communication is exchanged, and a second round of operations on each side are performed. If carried out correctly, any local computation can be reproduced in this non-local form. See also \cref{eq:non-localform}. (c) A view of the bulk of AdS. Quantum systems $A_0,A_1$ travel inward from $\mathcal{C}_0$ and $\mathcal{C}_1$ to the region $J_{01\rightarrow 01}$, where they interact under the channel $\mathcal{N}_{A_0A_1\rightarrow B_0B_1}$. Systems $B_0$, $B_1$ travel outward to $r_0$ and $r_1$. (d) The same set-up viewed in the boundary, showing the regions $\hat{J}_{i\rightarrow 01}$ and $\hat{J}_{01\rightarrow i}$. Because there is no location where $A_0$ and $A_1$ can meet while still being in the past of $r_0,r_1$, the channel implemented in the bulk must be reproduced in the non-local form shown above.}
\label{fig:cylinder1}
\end{center}
\end{figure*}
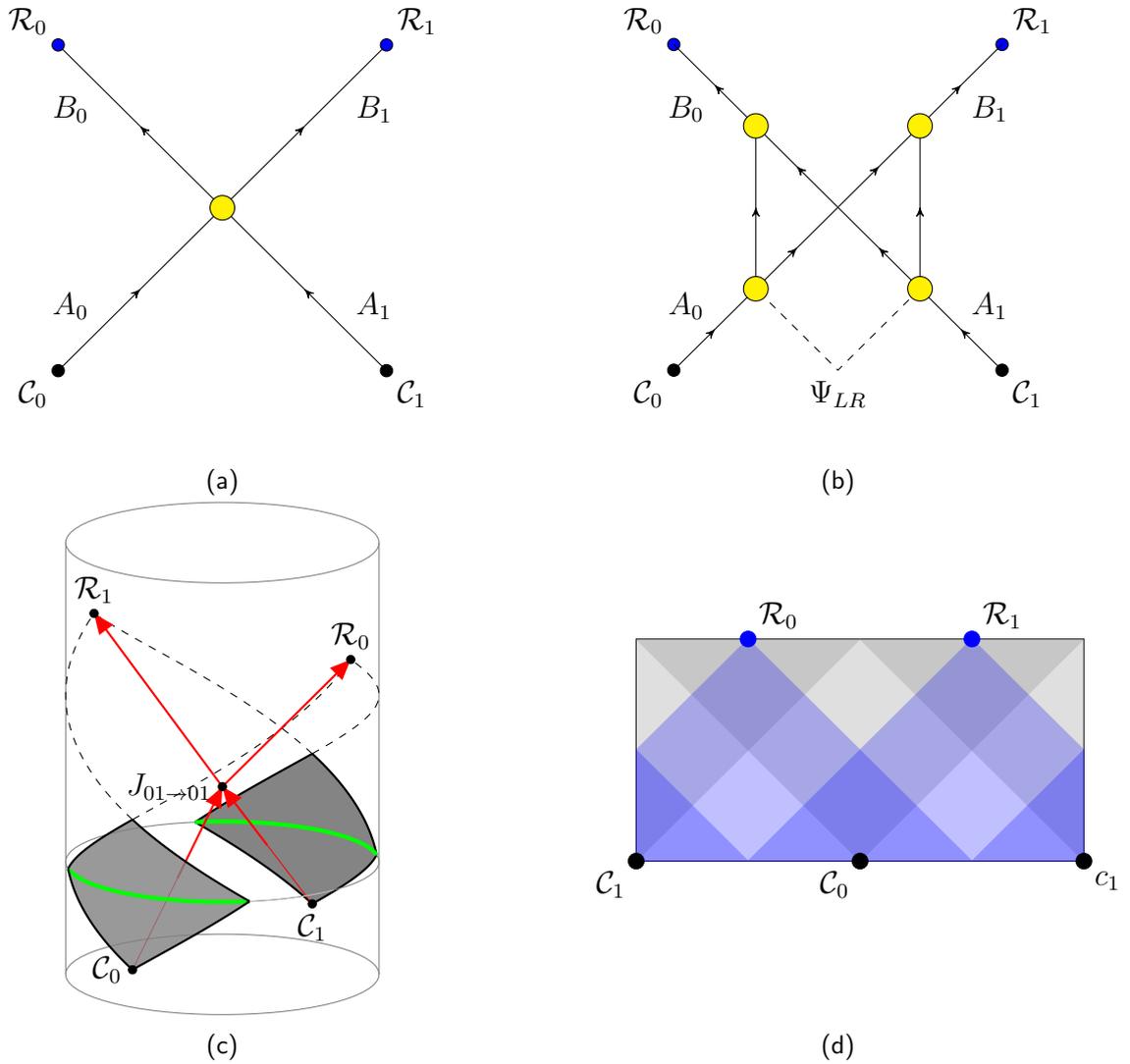

In this article, we will be particularly interested in the limits of computation within a finite gravitating region, and ask if gravity limits computation within such a region.
Our starting observation is that computations which happen within a finite bulk region are reproduced in a particular, non-local, form in the boundary. 
See \cref{fig:cylinder1}. 
This was pointed out in \cite{may2019quantum}, and the geometric observation underlying this has been made earlier \cite{heemskerk2009holography,gary2009local,maldacena2017looking}. 
For carefully chosen computations, it is possible to prove that entanglement between appropriate spacetime regions is necessary for the computation to be reproduced in this non-local form. 
In \cite{may2019quantum,may2020holographic,may2021holographic}, this was exploited to argue for constraints on the state of the CFT in the boundary that follow from aspects of bulk geometry. 
In particular, the non-emptiness of an overlap of a particular set of bulk light cones and an assumption that certain simple computations can happen within this region implies appropriate entanglement is available in the boundary. 
Here we reverse this perspective, and ask how the existence of the boundary description, and in particular the finite entanglement available there, constrains computation in the bulk.

To do this, we use the techniques in \cite{may2020holographic,may2021holographic} to make the connection between local computation in the bulk and non-local computation in the boundary quantitative. 
For a special class of bulk regions called \emph{scattering regions} constructed within pure AdS, we recall arguments from \cite{may2019quantum,may2020holographic,may2021holographic} that computations inside the region are reproduced in the boundary using an entangled system with mutual information equal to the area of the \emph{ridge}, a particular surface on the scattering region.\footnote{For arbitrary regions constructed in pure AdS, we can apply our arguments whenever that region can be enclosed within some scattering region.}
Existing methods for performing computations in the non-local form of \cref{subfig:nonlocalschematic} use entanglement which is exponential in the input size to implement a general unitary. 
Unless this can be dramatically improved, the fact that local computations in the bulk are reproduced using only an area's worth of entanglement will place some form of constraint on computation in the presence of gravity. 
To understand this constraint better, we wish to identify the property of a unitary that controls the entanglement cost to implement it non-locally. 
We explore the possibility that it is the complexity\footnote{There are some subtleties in which measure of complexity to use, we make the meaning of the complexity $\mathscr{C}$ appearing here more precise in the main article.} which controls the entanglement cost.
Given this conjecture, it would follow that the complexity of computations happening within such a region is bounded above by a function of the region's area. 

Non-local quantum computation is also of independent interest in quantum information theory, in particular because of the connection to position verification \cite{kent2006tagging,kent2011quantum,malaney2010location,buhrman2014position}. 
Some of our results are of independent interest in that context, including some new lower bounds on entanglement use. 
Further, a negative resolution of our conjecture --- that complexity controls entanglement cost --- would have important implications for position-verification. 
We discuss this further in section \ref{sec:discussion}.

\textbf{Overview of the paper:} In \cref{sec:preliminaries}, we establish a number of preliminaries needed to understand our arguments. 
In \cref{sec:localtononlocal} we review the argument of \cite{may2019quantum} that local computations in the bulk are reproduced in the boundary in the non-local form of \cref{subfig:nonlocalschematic}. 
In \cref{sec:scatteringgeometry} we highlight some observations from \cite{may2020holographic,may2021holographic} that relate the entanglement available in the boundary to the geometry of the bulk scattering region. 
In \cref{sec:entangledpart} we define a family of measures of complexity which are appropriate for our setting. 
Finally in \cref{sec:QGinsideNLQC} we state precisely the expected relationship between computations in a bulk scattering region and non-local computation in the boundary.

In \cref{sec:linearbounds} we begin discussing constraints on bulk computation that arise from the boundary perspective. 
We point out that the strongest known constraints on non-local computation lead to an upper bound on the input size to a bulk computation, enforcing in particular that the input size should be at most of the order of the area of the ridge. 
We point out that this constraint is always weak enough to be implied by the covariant entropy bound \cite{bousso1999covariant,flanagan2000proof, bousso2014proof}, applied in the bulk.

In \cref{sec:whatisthedual?} we motivate and state two conjectures relating entanglement cost and complexity. The weak complexity-entanglement conjecture states that complexity is the quantity controlling entanglement cost in non-local computation. The strong complexity-entanglement conjecture says that the entanglement cost is a polynomial in circuit complexity. We motivate this from a bulk perspective, noting that a simple assumption about bulk computation would imply this complexity-entanglement relationship.

In \cref{sec:upperandlowerbounds} we begin giving evidence from a quantum information perspective that the complexity of a local operation controls the entanglement cost to implement it non-locally. 
To start, in sections \ref{sec:speelmanbounds} and \ref{sec:PTbounds} we review some existing non-local computation protocols, which have an entanglement cost upper bounded by complexity. 
Then in \cref{sec:Clowerbound} we present a new (but weak) lower bound on entanglement in non-local computation, which is in terms of the complexity. 

In \cref{sec:restrictedprotocols} we discuss non-local computation under restrictions on the set of allowed protocols used. 
For certain classes of protocols we can prove upper and lower bounds on entanglement in terms of complexity that agree up to polynomial overheads.
In particular we prove this for protocols which use only Clifford unitaries.
We also point out an analogous result, applying to protocols restricted to Bell measurements and classical computations, proven earlier in \cite{buhrman2013garden}.

In \cref{sec:froutingprediction} we point out that the weak complexity-entanglement conjecture has an interesting consequence for the $f$-routing task, a well studied setting in the cryptographic literature.
In particular the conjecture implies more efficient protocols should exist for $f$-routing than were previously known. 
Following up on this reasoning, in recent work we established new protocols with the needed efficiency, supporting the conjecture \cite{cree2022code}. 

In \cref{sec:discussion} we summarize our argument that complexity controls entanglement in non-local computation. 

\Cref{sec:futuredirections} gives some future directions. 

\vspace{0.2cm}
\textbf{Summary of notation:}
\begin{itemize}
    \item We use script capital letters for boundary spacetime regions, $\mathcal{A}, \mathcal{B}, \mathcal{C}...$
    \item The entanglement wedge of a boundary region $\mathcal{A}$ is denoted by $E_{\mathcal{A}}$. 
    \item The Ryu-Takayanagi surface associated to region $\mathcal{A}$ is denoted $\gamma_{\mathcal{A}}$.
    \item We use $J^\pm(\mathcal{A})$ to denote the causal future or past of region $\mathcal{A}$ taken in the bulk geometry, and $\hat{J}^\pm(\mathcal{A})$ when restricting to the boundary geometry. 
    \item We use capital letters to denote quantum systems $A,B,C,...$
    \item We use boldface, script capital letters for quantum channels, $\mathbfcal{N}(\cdot)$, $\mathbfcal{T}(\cdot)$,...
    \item We use boldface capital letters to denote unitaries or isometries, $\mathbf{U}, \mathbf{V},...$
\end{itemize}

\section{Preliminaries}\label{sec:preliminaries}

\subsection{Holography implements non-local computations} \label{sec:localtononlocal}

The AdS/CFT correspondence describes quantum gravity in an asymptotically AdS$_{d+1}$ dimensional spacetime in terms of a lower dimensional, non-gravitating theory living at the AdS boundary. 
This raises a puzzle in the context of a geometric observation made in \cite{heemskerk2009holography,gary2009local,maldacena2017looking} which we illustrate for AdS$_{2+1}$ in \cref{subfig:cylinder}. 
Choose four points in the boundary, two at an early time and two at a later time.
We can choose these four points such that light rays from the two early points may meet, then travel from the meeting point to both late time points. 
Consider two systems $A_0$ and $A_1$, falling inward along these incoming light rays, interacting where they meet, then traveling outward to the boundary after the interaction has occurred. 
This process happens in the bulk picture under evolution by a local Hamiltonian. 
In the boundary picture, it can happen that for the same four points there is no region of the boundary where the two systems can be brought together while still in the past of both late time points. 
See \cref{subfig:boundary}. 
It then becomes puzzling how, using local Hamiltonian evolution, the boundary can reproduce the bulk dynamics. 

To resolve this puzzle, \cite{may2019quantum} argued that a local interaction in the bulk is described via a ``non-local quantum computation''\footnote{Our use of ``non-local'' here can be confusing. The boundary Hamiltonian is local in the sense that it respects causality. The boundary implementation of bulk interactions is ``non-local'' in that it uses the form of \cref{eq:non-localform} to implement the channel.} in the boundary. In particular a channel $\mathbfcal{N}_{A_0A_1\rightarrow B_0B_1}(\cdot)$ can always be implemented in the form \cite{buhrman2014position,beigi2011simplified}
\begin{align}\label{eq:non-localform}
    \mathbfcal{N}_{A_0A_1\rightarrow B_0B_1}(\cdot) &= (\mathbfcal{C}_{A_0'R'\rightarrow B_0}\otimes \mathbfcal{C}_{A_1'L'\rightarrow B_1})\circ \nonumber \\
    &(\mathbfcal{B}_{A_0L\rightarrow A_0'L'}\otimes \mathbfcal{B}_{A_1R\rightarrow A_1'R'})(\cdot \otimes \Psi_{LR}),
\end{align}
although in general it is not well understood how much entanglement is needed in the $LR$ system, as we discuss in detail later on.
The basic observation of \cite{may2019quantum} is that when a local channel is implemented in the bulk, it is realized in the non-local form on the right of \cref{eq:non-localform} in the boundary. 
The puzzle is resolved by the observation that local interactions plus shared entanglement in the more restrictive boundary geometry can be used to reproduce local interactions in the higher dimensional bulk. 
In fact, even general scenarios, with more inputs and outputs and requiring computations in multiple spacetime locations from a bulk perspective, can always be reproduced in the boundary with sufficient entanglement \cite{dolev2019constraining}. 

To understand the role of non-local computation in holography in more detail, we first need to introduce a few aspects of AdS/CFT. 
The basic claim of the AdS/CFT correspondence is that any bulk quantity can be calculated using boundary degrees of freedom, and vice versa. 
One important quantity studied in the boundary is the von Neumann entropy.
In CFT states corresponding to semi-classical geometries, this can be calculated in a simple way using the bulk via the Ryu-Takayanagi formula \cite{ryu2006aspects,hubeny2007covariant,lewkowycz2013generalized,faulkner2013quantum,dong2016deriving,engelhardt2015quantum,dong2018entropy}\footnote{See also \cite{akers2021leading,dong2020enhanced,marolf2020probing} for conditions on when this formula holds. In particular, \cite{akers2021leading} showed this can break down when bulk matter is in an `incompresible' state, while \cite{dong2020enhanced,marolf2020probing} found corrections near the transition between different extremal surfaces.}, 
\begin{align}
    S(\mathcal{R}) = \text{ext}_{\gamma\in \text{Hom}(\mathcal{R})} \left( \frac{\text{area}(\gamma)}{4G_N} + S_{bulk}(E_\gamma') \right).
\end{align}
The extremization is taken over all spacelike, codimension 2 surfaces $\gamma$ which are homologous to $\mathcal{R}$, which means that there exists some codimension 1 spacelike surface $E_\gamma'$ such that $\partial E_{\gamma}'=\gamma\cup \mathcal{R}$. 
The term $S_{bulk}(E_\gamma')$ captures the entropy of degrees of freedom sitting within $E_{\gamma}'$. 
If there is more than one extremal surface, one should pick the surface that minimizes the functional on the right of the above equation.
There is also a useful reformulation of this formula known as the \emph{maximin formula} \cite{wall2014maximin,akers2020quantum}, stated as follows. 
\begin{align}
    S(\mathcal{R}) = \max_{\Sigma} \min_{\gamma\in \text{Hom}(\mathcal{R})} \left( \frac{\text{area}(\gamma)}{4G_N} + S_{bulk}(E_\gamma') \right).
\end{align}
Here the maximization is over bulk Cauchy surfaces which include $\mathcal{R}$ in their boundary. 

The Ryu-Takayanagi formula picks a minimal extremal codimension 2 surface, call it $\gamma_\mathcal{R}$, given the boundary region $\mathcal{R}$. 
We call this the \emph{RT surface} of $\mathcal{R}$. 
It also picks out a family of codimension one surfaces $E_\mathcal{R}'$ such that $\partial E_\mathcal{R}' = \gamma_\mathcal{R} \cup \mathcal{R}$. 
Taking the domain of dependence of any member of this family we obtain a codimension $0$ region we label $E_\mathcal{R}$ and call the \emph{entanglement wedge}. 
In this section we will be most interested in the entanglement wedge, in particular because it has an important meaning as the portion of the bulk spacetime which is recorded into the density matrix $\rho_R$ \cite{czech2012gravity,headrick2014causality,wall2014maximin,jafferis2016relative,dong2016reconstruction,cotler2019entanglement,akers2019large,akers2021leading}. 
In particular, a quantum system $Q$ which in the bulk picture can be recovered from the region $E_{\mathcal{R}}$ can in the boundary picture be recovered from region $\mathcal{R}$. This statement is known as \emph{entanglement wedge reconstruction}.

The next step towards understanding \cref{eq:non-localform} will be to introduce the ``quantum tasks'' language of \cite{kent2011quantum,may2019quantum,may2021holographic}. 
In a quantum task, we consider two parties, Alice and Bob. 
Bob prepares a set of input systems, which he gives to Alice at a set of input locations. 
In particular at input location $\mathcal{C}_i$ Bob gives Alice system $A_i$. 
Alice should process the input systems in some way, then return quantum systems $B_i$ at a set of output locations $\mathcal{R}_i$. 
To fully specify the task we need to specify an initial state of the boundary CFT, which fixes the resource systems available in the boundary to perform the task, as well as fixes the bulk geometry and fields.
We call this initial state $\ket{\Psi}$. 

While the early references considering quantum tasks in the context of AdS/CFT \cite{may2019quantum,may2020holographic} deal with the case where the input and output locations consist of points, it will be helpful to follow \cite{may2021holographic} and allow input and output systems to be recorded into spacetime regions. 
Because of entanglement wedge reconstruction, a quantum task defined in the bulk with input regions $E_{\mathcal{C}_i}$ and output regions $E_{\mathcal{R}_i}$ corresponds in the boundary picture to a task with the same input and output systems, but now with input regions $\mathcal{C}_i$ and output regions $\mathcal{R}_i$.

In a quantum task then, we have Bob couple to the CFT within the regions $\mathcal{C}_i$ to introduce the input systems, and then couple to $\mathcal{R}_i$ to collect the output systems. 
We can take two perspectives on the role of Alice. The first is what we will call the `Alice-in-the-box' perspective, where Alice is thought of as a collection of agents living within the AdS spacetime.
In this case, Bob is holding a CFT and probing it, and all of Alice's actions are recorded already into the initial CFT state $\ket{\Psi}$. 
This is the perspective used earlier \cite{may2019quantum,may2020holographic,may2021holographic}. 
Here, we will additionally make use of an alternative `Alice-out-of-the-box' perspective, wherein Alice is, like Bob, a set of agents outside the CFT, who may couple to it. 
We require however that this `out-of-the-box' Alice 1) respect the causal structure of the boundary spacetime 2) not hold entanglement aside from the CFT degrees of freedom.
Now we are ready to understand \cref{eq:non-localform}. Consider a quantum task which has only two input regions and two output regions. 
That is, Bob records systems $A_0$, $A_1$ into $E_{\mathcal{C}_0}$, $E_{\mathcal{C}_1}$, and requires systems $B_0$, $B_1$ be recorded into regions $E_{\mathcal{R}_0}$, $E_{\mathcal{R}_1}$.
We require the input and output systems be related by some specified quantum channel $\mathbfcal{N}_{A_0A_1\rightarrow B_0B_1}$. 
We show this in \cref{subfig:cylinder}, where the input and output regions are located in the boundary of an asymptotically AdS spacetime. 
We are most interested in the setting where $\mathcal{C}_0,\mathcal{C}_1,\mathcal{R}_0,\mathcal{R}_1$ are such that the region\footnote{Here $J^+(\mathcal{X})$ indicates the causal future of the region $\mathcal{X}$, and $J^-(\mathcal{X})$ indicates the causal past.}
\begin{align}\label{eq:scatteringregion}
    J_{01\rightarrow 01} \defi J^+(E_{\mathcal{C}_0}) \cap J^+(E_{\mathcal{C}_1}) \cap J^-(E_{\mathcal{R}_0}) \cap J^-(E_{\mathcal{R}_1})
\end{align}
is non-empty, and such that this region does not extend to the AdS boundary. We call this the \emph{scattering region}. 

In the bulk perspective, the task can be completed by bringing $A_0, A_1$ into the scattering region, performing the channel, and sending the outputs to $E_{\mathcal{R}_0}$, $E_{\mathcal{R}_1}$. 
Now consider this same process from a boundary perspective. Because the scattering region does not extend to the boundary, there is no spacetime region where $A_0$ and $A_1$ can be directly interacted to perform the channel. 
Nonetheless, while only evolving according to its local Hamiltonian, the boundary theory has to reproduce the overall channel. What are available in the boundary are the regions
\begin{align}
    \hat{J}_{i\rightarrow 01} &= \hat{J}^+(\mathcal{C}_i) \cap \hat{J}^-(\mathcal{R}_0) \cap \hat{J}^-(\mathcal{R}_1), \nonumber \\
    \hat{J}_{01\rightarrow i} &= \hat{J}^-(\mathcal{R}_i) \cap \hat{J}^+(\mathcal{C}_0) \cap \hat{J}^+(\mathcal{C}_1),
\end{align}
where the $\hat{J}^\pm$ indicate future and past light cones taken in the boundary geometry. We call $\mathcal{V}_i\defi \hat{J}_{i\rightarrow 01}$ the \emph{decision regions}, and also define $\mathcal{W}_i\defi \hat{J}_{01\rightarrow i}$. 
Performing quantum channels within each of these regions on the available systems, including entanglement shared between $\mathcal{V}_0$, $\mathcal{V}_1$, is exactly a computation in the form of \cref{eq:non-localform}. 
We show this in \cref{subfig:nonlocalschematic}. 

The shared correlation in the non-local computation circuit of \cref{subfig:nonlocalschematic} is, in the spacetime picture, held between regions $\mathcal{V}_0$ and $\mathcal{V}_1$. 
Because these are separated regions, the correlation between them in the conformal field theory is finite.
The basic strategy we pursue in this paper is to try and constrain the computations which can happen inside the scattering region by constraining which computations can be performed non-locally in the boundary picture with this limited correlation.

There is one possible obstruction in identifying what happens in the boundary CFT with a non-local computation in the form of \cref{eq:non-localform}. 
In the idealized form of \cref{eq:non-localform}, the output locations $\mathcal{R}_0$ and $\mathcal{R}_1$ receive quantum systems from $\mathcal{C}_0$ and $\mathcal{C}_1$, and nowhere else. 
In the holographic setting, the corresponding locations $\mathcal{W}_0$, $\mathcal{W}_1$ additionally have some region, call it $\mathcal{X}$, sitting in their past light cones but outside the future of both $\mathcal{C}_0$ and $\mathcal{C}_1$. 
In particular $\mathcal{X}$ is the spacelike complement of $\mathcal{V}_0 \cup \mathcal{V}_1$.
Using this extra region, \cite{may2020holographic} showed how to complete certain non-local computations with zero mutual information between regions $\mathcal{V}_0$ and $\mathcal{V}_1$, even while this mutual information is provably necessary when the region $\mathcal{X}$ is absent. 
In fact, it is now known that any computation can be performed non-locally with zero mutual information between the input regions by exploiting the $\mathcal{X}$ regions \cite{dolev2022holography}!
Given this, it does not seem reasonable to take the causal structure implicit in \cref{eq:non-localform} as our model for the boundary. 

However, in this article we are interested in asymptotics, in particular in how the class of computations that can be performed within the scattering region behaves as the size of the region is made large. 
For instance, we will ask about how the complexity of computations that can happen within the region scales with the area of the ridge at large ridge area. 
For this purpose, we can better justify the form \cref{eq:non-localform} by noting that to make the ridge large, we are moving the output points forwards in time, and in the boundary this corresponds to the $\mathcal{X}$ regions shrinking to empty.
Thus our assumption is that as the $\mathcal{X}$ regions shrink and become empty, their relevance to the class of computations which can happen in the scattering region becomes small.

Additionally, there are hints that the causal structure shown in \cref{subfig:nonlocalschematic} is a good model for how bulk computations are reproduced in the boundary, even for fixed, large $X$ regions (although we emphasize that we don't need this). 
First, the gravity proof of the connected wedge theorem (see theorem \ref{thm:CW}) shows that the mutual information $I(\mathcal{V}_0:\mathcal{V}_1)$ is $\Theta(1/G_N)$ whenever the scattering region exists.\footnote{Note that we use computer scientists conventions in our use of big $\Theta$ and big $O$ notation: $f(x)=O(g(x))$ means $\lim_{x\rightarrow \infty} f(x)/g(x)<\infty $, while $f(x)=\Theta(g(x))$ means $0< \lim_{x\rightarrow \infty} f(x)/g(x) < \infty$.}
This follows from a boundary argument if we naively assume \cref{subfig:nonlocalschematic} is a good model, but does not follow if we include the $\mathcal{X}$ regions, suggesting the AdS/CFT correspondence fails to exploit the $\mathcal{X}$ regions for some (apparently mysterious) reason. 
Another hint is that the scattering region sits inside of the entanglement wedge of the input regions, again suggesting the primacy of correlation between the input regions in supporting bulk computation. 
We leave better understanding these observations to future work. 

\subsection{Geometry of the scattering region and boundary entanglement}\label{sec:scatteringgeometry}

\begin{figure}
\centering
\tdplotsetmaincoords{15}{0}
\begin{tikzpicture}[scale=1.05,tdplot_main_coords]
\tdplotsetrotatedcoords{0}{35}{0}
    
    \begin{scope}[tdplot_rotated_coords]
    
    \draw[thick,dashed] (-3,0,0) -- (0,4,3);
    
    \draw[opacity=0.25,fill=red] (0,4,3) -- (-3,0,0) -- (3,0,0) -- (0,4,3);
    \draw[opacity=0.25,fill=blue] (-3,0,0) -- (0,4,3) -- (0,4,-3) -- (-3,0,0);
    \draw[opacity=0.6,fill=blue] (3,0,0) -- (0,4,3) -- (0,4,-3) -- (3,0,0);
    \draw[opacity=0.7,fill=red] (0,4,-3) -- (-3,0,0) -- (3,0,0) -- (0,4,-3);
    
    \draw[thick] (-3,0,0) -- (3,0,0);
    \draw[thick] (0,4,-3) -- (0,4,3);
    
    \draw[thick] (-3,0,0) -- (0,4,-3);
    
    \draw[thick] (3,0,0) -- (0,4,-3);
    \draw[thick] (3,0,0) -- (0,4,3);
    
    \node[below] at (0,-0.25,0) {$r$};
    \draw[->] (0.25,-0.5,0) to [out=0,in=-90] (0.5,0,0);
    
    \end{scope}
    \end{tikzpicture}
    \caption{A scattering region in AdS$_{2+1}$. The lower edge is the ridge, $r$. The red faces are $\Sigma^0$ and $\Sigma^1$, the upper faces are $\Gamma^0$ and $\Gamma^1$.}
    \label{fig:scatteringregion}
\end{figure}
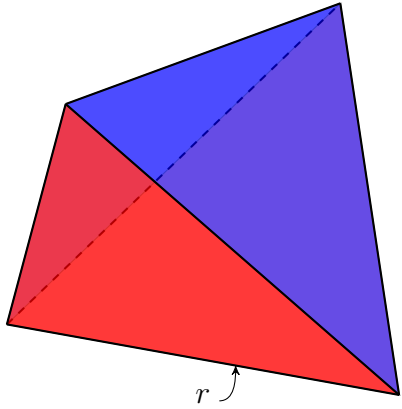

It will be helpful to have a more detailed picture of the geometry of the scattering region, and to understand something about how this geometry relates quantitatively to entanglement in the boundary.
First, recall the definition of the scattering region from \cref{eq:scatteringregion},
\begin{align}
    J_{01\rightarrow 01} \defi J^+(E_{\mathcal{C}_0}) \cap J^+(E_{\mathcal{C}_1}) \cap J^-(E_{\mathcal{R}_0}) \cap J^-(E_{\mathcal{R}_1}). \nonumber 
\end{align}
We illustrate a scattering region constructed in $2+1$ bulk dimensions in figure \cref{fig:scatteringregion}. 
The scattering region takes the basic shape of a tetrahedron, though the sides and faces can be curved. 
The lower two faces (shown in red) are portions of the boundary of the future of $\mathcal{C}_0$ and $\mathcal{C}_1$. In particular, we define
\begin{align}\label{eq:lowerfaces}
    \Sigma^0 &\defi \partial J^{+}(E_{\mathcal{C}_0}) \cap J^-(E_{\mathcal{R}_0}) \cap J^-(E_{\mathcal{R}_1}), \nonumber \\
    \Sigma^1 &\defi \partial J^{+}(E_{\mathcal{C}_1}) \cap J^-(E_{\mathcal{R}_0}) \cap J^-(E_{\mathcal{R}_1}).
\end{align}
The upper two faces (shown in blue) are portions of the boundary of the past of $\mathcal{R}_0$ and $\mathcal{R}_1$, 
\begin{align}\label{eq:upperfaces}
    \Gamma^0 &\defi \partial J^-(E_{\mathcal{R}_0}) \cap J^+(E_{\mathcal{C}_0}) \cap J^+(E_{\mathcal{C}_1}), \nonumber \\ 
    \Gamma^1 &\defi \partial J^-(E_{\mathcal{R}_1}) \cap J^+(E_{\mathcal{C}_0}) \cap J^+(E_{\mathcal{C}_1}) .
\end{align}
Finally, we will be interested in the lower edge of the tetrahedron, which we call the ridge
\begin{align}\label{eq:ridge}
    r \defi \partial J^{+}(E_{\mathcal{C}_0}) \cap \partial J^{+}(E_{\mathcal{C}_1}) \cap J^-(E_{\mathcal{R}_0}) \cap J^-(E_{\mathcal{R}_1}),
\end{align}
which is also just $\Sigma^0 \cap \Sigma^1$. 

As mentioned in the last section, whenever the scattering region is non-empty the mutual information $I(\mathcal{V}_0:\mathcal{V}_1)$ is large, and in particular it is $\Theta(1/G_N)$. 
In fact, \cite{may2020holographic, may2021holographic} proved the following theorem.\footnote{Note that the theorem is not stated in quite the following way in those references, but this statement follows from that proof.}
\begin{theorem}\label{thm:CW} \textbf{(Connected wedge theorem)} Pick four regions ${\mathcal{C}}_0,{\mathcal{C}}_1,{\mathcal{R}}_0,{\mathcal{R}}_1$ on the boundary of an asymptotically AdS spacetime such that
\begin{align}
    \mathcal{C}_0 = \mathcal{V}_0 &\defi \hat{J}^+({\mathcal{C}}_0) \cap \hat{J}^-({\mathcal{R}}_0) \cap \hat{J}^-({\mathcal{R}}_1) , \nonumber \\
    \mathcal{C}_1 = \mathcal{V}_1 &\defi \hat{J}^+({\mathcal{C}}_1) \cap \hat{J}^-({\mathcal{R}}_0) \cap \hat{J}^-({\mathcal{R}}_1).
\end{align}
Then 
\begin{align}\label{eq:mutualinfoandarea}
    \frac{1}{2}I(\mathcal{V}_0:\mathcal{V}_1) \geq \frac{\text{area}(r)}{4G_N} + O(1).
\end{align}
with the ridge $r$ defined as in \cref{eq:ridge}. 
\end{theorem}
It will be instructive to briefly outline the elements of the proof of \cref{thm:CW}. 
For more details the reader may consult \cite{may2021holographic}.
Begin by recalling that the mutual information is
\begin{align}
    I(\mathcal{V}_0:\mathcal{V}_1) = \frac{1}{4G_N}\left(\text{area}(\gamma_{\mathcal{V}_0}\cup \gamma_{\mathcal{V}_1})- \text{area}(\gamma_{\mathcal{V}_0\cup \mathcal{V}_1}) \right) + O(1),
\end{align}
where we assume the bulk matter has $O(1)$ entropy. 
We will use a focusing argument to compare the sum of the first term to the the second. 
We consider the \emph{lift}
\begin{align}
    \mathcal{L}=\partial J^+(\mathcal{V}_0\cup \mathcal{V}_1) \cap J^-(\mathcal{R}_0) \cap J^-(\mathcal{R}_1).
\end{align}
As well, take a Cauchy surface $\Sigma$ that contains $\gamma_{\mathcal{V}_0\cup \mathcal{V}_1}$, and define the \emph{slope}, 
\begin{align}
    \mathcal{S}_\Sigma = \partial [J^-(\mathcal{R}_0) \cup J^-(\mathcal{R}_1)] \cap J^-[\partial J^+(\mathcal{V}_0\cup \mathcal{V}_1)] \cap J^+(\Sigma).
\end{align}
In \cref{fig:nullmembrane} we illustrate the \emph{null membrane}, which is the union of the lift and the slope. 
The null membrane allows us to compare the area of $\gamma_{\mathcal{V}_0}\cup \gamma_{\mathcal{V}_1}$ to the area of $\gamma_{\mathcal{V}_0\cup \mathcal{V}_1}$. 
To see this, note that light rays moving up the lift have decreasing area, because they are initially orthogonal to $\gamma_{\mathcal{V}_0}\cup \gamma_{\mathcal{V}_1}$. 
Similarly, light rays moving down the slope have decreasing area, because they are initially orthogonal to $\gamma_{\mathcal{V}_0\cup \mathcal{V}_1}$. 
Thus, consider pushing $\gamma_{\mathcal{V}_0}\cup \gamma_{\mathcal{V}_1}$ upwards along the lift, removing any light rays that collide with each other and stopping whenever the surface reaches the slope.
Note that we remove at least twice the ridges area in doing this. Then push the surface downwards along the slope, again noting that the area decreases, stopping when we reach $\Sigma$. 
This produces a surface $C_\Sigma$ in the slice $\Sigma$ which has area smaller than $\gamma_{\mathcal{V}_0}\cup \gamma_{\mathcal{V}_1}$ by at least $2\,\text{area}(r)$. 
As well, since $C_\Sigma$ lies in a single Cauchy slice with $\gamma_{\mathcal{V}_0\cup \mathcal{V}_1}$, according to the maximin formula it must have area larger or equal to that of $\gamma_{\mathcal{V}_0\cup \mathcal{V}_1}$. 
This gives
\begin{align}\label{eq:mutualinfoandridge}
    I(\mathcal{V}_0:\mathcal{V}_1) &\geq \frac{1}{4G_N}\left(\text{area}(\gamma_{\mathcal{V}_0}\cup \gamma_{\mathcal{V}_1})- \text{area}(C_\Sigma) \right)+O(1) \nonumber \\
    &\geq 2 \, \text{area}(r) + O(1)
\end{align}
which is as needed. 

\begin{figure}
    \centering
    \tdplotsetmaincoords{15}{0}
    \begin{tikzpicture}[scale=2.2,tdplot_main_coords]
    \tdplotsetrotatedcoords{0}{30}{0}
    
    \begin{scope}[tdplot_rotated_coords]
    
    \draw[domain=40:140,variable=\x,smooth,ultra thick,blue] plot ({2*sin(\x+180)+2.53}, {0}, {2*cos(\x+180)});
    \draw[dashed,domain=40:140,variable=\x,smooth,ultra thick,blue] plot ({-(2*sin(\x+180)+2.53)}, {0}, {2*cos(\x+180)});
    
    \draw[dashed,domain=-40:40,variable=\x,smooth,ultra thick,red] plot ({-(2*sin(\x+180))}, {0}, {2*cos(\x+180)+3.1});
    \draw[domain=-40:40,variable=\x,smooth,ultra thick,red] plot ({-(2*sin(\x+180))}, {0}, {-(2*cos(\x+180)+3.1)});
    
    \fill[red,opacity=0.4,domain=40:-40] (-1.286,0,1.532) -- (0,0.5,0.5) -- (1.286,0,1.532) plot ({-(2*sin(\x+180))}, {0}, {2*cos(\x+180)+3.1});
    
    \fill[blue,opacity=0.4,domain=40:140] (1.286,0,1.532) -- (0,0.5,0.5) -- (0,0.5,-0.5) -- (1.286,0,-1.532) plot ({2*sin(\x+180)+2.53}, {0}, {2*cos(\x+180)});
    \draw[line width=1mm, white] (1.27,0,-1.5) -- (1.27,0,1.5);
    
    \fill[blue,opacity=0.4,domain=40:140] (-1.286,0,1.532) -- (0,0.5,0.5) -- (0,0.5,-0.5) -- (-1.286,0,-1.532) plot ({-(2*sin(\x+180)+2.53)}, {0}, {2*cos(\x+180)});
    \draw[line width=1mm, white] (-1.27,0,-1.5) -- (-1.27,0,1.5);
    
    \fill[red,opacity=0.4,domain=40:-40] (-1.286,0,-1.532) -- (0,0.5,-0.5) -- (1.286,0,-1.532) plot ({-(2*sin(\x+180))}, {0}, {-(2*cos(\x+180)+3.1)});
    \draw[ultra thick,white] (-1.286,0,-1.537) -- (1.286,0,-1.537);
    
    \draw[domain=0:360,thick,gray] plot ({2*cos(\x)},0, {2*sin(\x)});
    
    \draw[black,thick] (1.286,0,1.532) -- (0,0.5,0.5);
    \draw[black,thick] (1.286,0,-1.532) -- (0,0.5,-0.5);
    \draw[black,thick] (-1.286,0,1.532) -- (0,0.5,0.5);
    \draw[black,thick] (-1.286,0,-1.532) -- (0,0.5,-0.5);
    
    \draw[thick] (0,0.5,-0.5) -- (0,0.5,0.5);
    
    \draw [green,ultra thick,domain=40:140] plot ({2*cos(\x-90)},0, {2*sin(\x-90)});
    \draw [green,ultra thick,domain=40:140] plot ({-2*cos(\x-90)},0, {2*sin(\x-90)});
    
    \node[above] at (0,0.75,0) {$r$};
    \node at (0.25,0,-1.5) {$C_\Sigma$};
    
    \node at (1,0,0) {$\mathcal{\gamma}_{\mathcal{V}_0}$};
    
    \node[right] at (2.2,0,0) {$\mathcal{V}_0$};
    \node[left] at (-2.2,0,0) {$\mathcal{V}_1$};

    \end{scope}
    \end{tikzpicture}
    \caption{The null membrane. The blue surface is the lift $\mathcal{L}$, which is generated by the null geodesics defined by the inward, future pointing null normals to $\gamma_{\mathcal{V}_0}\cup \gamma_{\mathcal{V}_1}$. The red surfaces make up the slope $\mathcal{S}_\Sigma$, which is generated by the null geodesics defined by the inward, past directed null normals to $\gamma_{{\mathcal{R}}_1}\cup \gamma_{{\mathcal{R}}_2}$. The ridge $r$ is where null rays from $\gamma_{{\mathcal{V}}_1}$ and $\gamma_{{\mathcal{V}}_2}$ collide. The contradiction surface $C_\Sigma$ is where the slope meets a specified Cauchy surface $\Sigma$. Figure reproduced from \cite{may2021holographic}.}
    \label{fig:nullmembrane}
\end{figure}
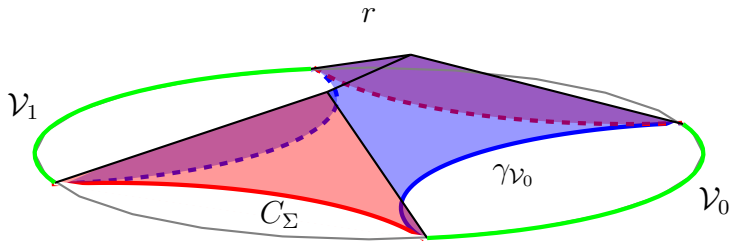

Notice from the argument in the last paragraph that if the expansion of the null membrane is everywhere zero, the only area removed during the focusing argument is the area of the ridge, so inequality \eqref{eq:mutualinfoandarea} is saturated. This occurs for example in pure AdS. 
Another comment on this theorem is that the converse is not true, in that there are states with $\Theta(1/G_N)$ mutual information but where the scattering region is empty. 

In the context of understanding the limits of computations in the bulk, we are interested in scenarios where our computation may change the bulk geometry, for instance because the inputs to the computation and matter used to build a computing device backreacts on the geometry. 
Fortunately, the way in which we have defined the scattering region remains natural in this setting. 
Recall that the definition of the quantum task includes a specification of the state of the CFT at some early time, which we call $\ket{\Psi}$. 
The scattering region is defined by $\ket{\Psi}$ along with a choice of input and output regions according to \cref{eq:scatteringregion}. 
Alice can then attempt to perform computations within the scattering region by coupling to the decision regions, which in general deforms the bulk geometry. 
The data defining the region, namely the choice of boundary regions plus initial state $\ket{\Psi}$, remains fixed however, so we can define a precise notion of the `same' scattering region even for the deformed geometries. 

We have defined our scattering region in terms of boundary data, but later will want to characterize scattering regions from a bulk perspective. 
Given this, we can ask what data about the scattering region remains well defined, even when Alice is allowed to couple to the decision region. 
We focus on the case where the initial state $\ket{\Psi}$ is the CFT vacuum. 
In this case, and without Alice coupling to the decision regions, we have that equation \ref{eq:mutualinfoandridge} is an equality and half the mutual information is equal to the ridge area. 
Noting that Alice's couplings don't change the mutual information, and applying theorem \ref{thm:CW}, we see that her couplings can only decrease the ridge area.\footnote{From a bulk perspective, this fact is related to the addition of bulk matter delaying light rays through the bulk, so that adding bulk matter will close off the scattering region.}  
For scattering regions defined with the initial state as the CFT vacuum, we can use the area of the ridge in the undeformed geometry to characterize the scattering region from a bulk perspective. 
Whenever we refer to the area of a scattering region $S$, we will mean the area of the ridge $r$ in the undeformed geometry. 

\subsection{The `interaction class' of a unitary}\label{sec:entangledpart}

We are interested in studying computations that can happen in the bulk scattering region, and in understanding this using the boundary description. 
Specifying the inputs and outputs as they appear near the boundary however does not precisely capture the computation which happens inside the scattering region. 
Instead, the bulk process can involve messages sent directly from $\mathcal{C}_0$ and $\mathcal{C}_1$ to $\mathcal{R}_0$ or $\mathcal{R}_1$, and can involve operations that act separately on $A_0$ and $A_1$ as they fall into or come out of the scattering region. 
To capture the computation which happens inside of the scattering region, we need a notion of the \emph{interaction class} of $\mathbf{U}$, which we define as follows. 

\begin{figure*}
    \centering
    \begin{subfigure}{0.45\textwidth}
    \centering
    \begin{tikzpicture}[scale=0.5]
    
    \draw[thick] (-5,-5) -- (-5,-3) -- (-3,-3) -- (-3,-5) -- (-5,-5);
    \node at (-4,-4) {$\mathbf{V}^L$};
    
    \draw[thick] (5,-5) -- (5,-3) -- (3,-3) -- (3,-5) -- (5,-5);
    \node at (4,-4) {$\mathbf{V}^R$};
    
    \draw[thick] (5,5) -- (5,3) -- (3,3) -- (3,5) -- (5,5);
    \node at (4,4) {$\mathbf{W}^R$};
    
    \draw[thick] (-5,5) -- (-5,3) -- (-3,3) -- (-3,5) -- (-5,5);
    \node at (-4,4) {$\mathbf{W}^L$};
    
    \draw[thick] (-4.5,-3) -- (-4.5,3);
    
    \draw[thick] (4.5,-3) -- (4.5,3);
    
    \draw[thick] (-3.5,-3) to [out=90,in=-90] (3.5,3);
    
    \draw[thick] (3.5,-3) to [out=90,in=-90] (-3.5,3);
    
    \draw[thick] (-3.5,-5) to [out=-90,in=-90] (3.5,-5);
    \draw[black] plot [mark=*, mark size=3] coordinates{(0,-7.05)};
    
    \draw[thick] (-4.5,-6) -- (-4.5,-5);
    \draw[thick] (4.5,-6) -- (4.5,-5);
    
    \draw[thick] (4.5,5) -- (4.5,6);
    \draw[thick] (-4.5,5) -- (-4.5,6);
    
    \draw[thick] (3.5,5) -- (3.5,6);
    \draw[thick] (-3.5,5) -- (-3.5,6);
    
    \end{tikzpicture}
    \caption{}
    \label{fig:non-local}
    \end{subfigure}
    \hfill
    \begin{subfigure}{0.45\textwidth}
    \centering
    \begin{tikzpicture}[scale=0.5]
    
    \draw[thick] (-5,-5) -- (-5,-3) -- (-3,-3) -- (-3,-5) -- (-5,-5);
    \node at (-4,-4) {$\mathbf{V}^L$};
    
    \draw[thick] (5,-5) -- (5,-3) -- (3,-3) -- (3,-5) -- (5,-5);
    \node at (4,-4) {$\mathbf{V}^R$};
    
    \draw[thick] (5,5) -- (5,3) -- (3,3) -- (3,5) -- (5,5);
    \node at (4,4) {$\mathbf{W}^R$};
    
    \draw[thick] (-5,5) -- (-5,3) -- (-3,3) -- (-3,5) -- (-5,5);
    \node at (-4,4) {$\mathbf{W}^L$};
    
    \draw[thick] (-4.5,-3) -- (-4.5,3);
    
    \draw[thick] (4.5,-3) -- (4.5,3);
    
    \draw[thick] (-4,-0.5) to [out=90,in=-90] (4,3);
    \draw[thick] (-4,-0.5) -- (-4,-3);
    
    \draw[thick] (4,0) to [out=90,in=-90] (-4,3);
    \draw[thick] (4,0) -- (4,-3);
    
    \draw[thick] (-4.5,-6) -- (-4.5,-5);
    \draw[thick] (4.5,-6) -- (4.5,-5);
    
    \draw[thick] (4.5,5) -- (4.5,6);
    \draw[thick] (-4.5,5) -- (-4.5,6);
    
    \draw[thick] (3.5,5) -- (3.5,6);
    \draw[thick] (-3.5,5) -- (-3.5,6);
    
    \draw[thick] (-1,-1) -- (-1,1) -- (1,1) -- (1,-1) -- (-1,-1);
    
    \draw[thick] (-3.5,-3) to [out=90,in=-90] (-0.5,-1);
    \draw[thick] (3.5,-3) to [out=90,in=-90] (0.5,-1);
    
    \draw[thick] (0.5,1) to [out=90,in=-90] (3.5,3);
    \draw[thick] (-0.5,1) to [out=90,in=-90] (-3.5,3);
    
    \node at (0,0) {$\mathbf{U}^I$};
    
    \end{tikzpicture}
    \caption{}
    \label{fig:entangledpart}
    \end{subfigure}
    \caption{(a) Circuit diagram showing the non-local implementation of a unitary $\mathbf{U}$. (b) Circuit diagram showing the implementation of a unitary in terms of a unitary $\mathbf{U}^I$ in the interaction class of $\mathbf{U}$, along with local pre- and post-processing. We are interested in the minimal complexity of unitaries in the interaction class, and how this relates to the number of EPR pairs needed in a non-local implementation of the same unitary, as shown at left. The isometries $\mathbf{V}^L,\mathbf{V}^R, \mathbf{W}^L, \mathbf{W}^R$ do not need to be the same in the two pictures.}
    \label{fig:entangledpartandnon-local}
\end{figure*}
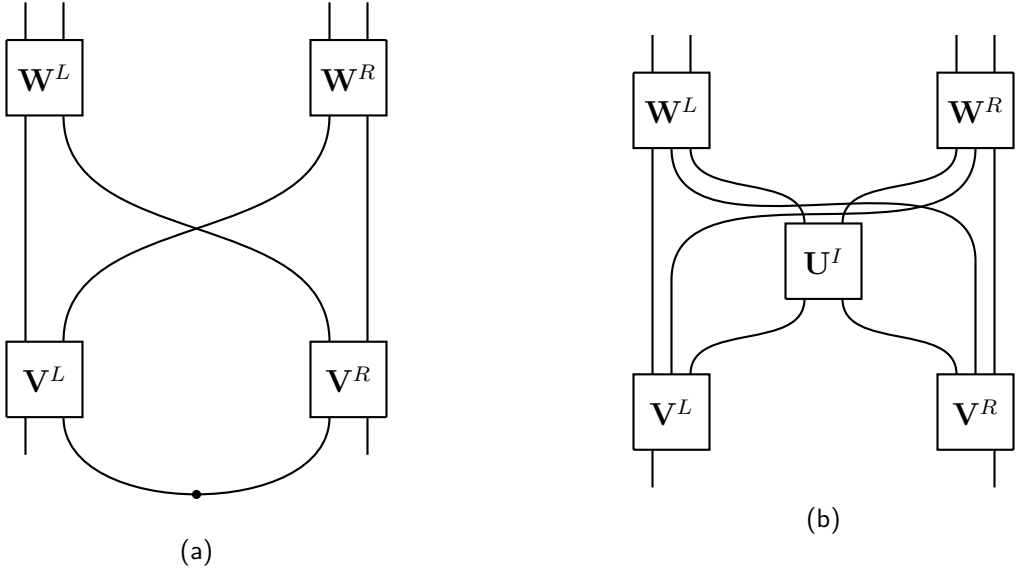

\begin{definition}
Given a unitary $\mathbf{U}$ the $\epsilon$ approximate \emph{interaction class} of $\mathbf{U}$ is the set of unitaries $\mathbf{U}^I$ such that there exist isometries $\mathbf{V}^L, \mathbf{V}^R, \mathbf{W}^L, \mathbf{W}^R$ satisfying
\begin{align}
    ||\mathbf{U}_{A_0A_1} - \tr_{Z_0Z_1}\left( [\mathbf{W}^L_{X_0Y_0A_0' \rightarrow A_0 Z_0} \otimes \mathbf{W}^R_{X_1Y_1A_1'\rightarrow A_1Z_1}]\mathbf{U}^I_{A_0'A_1'} [\mathbf{V}^L_{A_0\rightarrow A_0'X_0X_1} \otimes \mathbf{V}^R_{A_1\rightarrow A_1' Y_0Y_1}] \right)||_{\diamond} \leq \epsilon \nonumber 
\end{align}
See \cref{fig:entangledpart} for a circuit diagram. We denote the interaction class of $\mathbf{U}_{A_0A_1}$ by $\mathscr{I}_{A_0:A_1}(\mathbf{U})$. 
\end{definition}
We have used the diamond norm in the above definition so that the circuit exploiting an interaction unitary is operationally nearly indistinguishable from the target unitary. 
As an example, for a tensor product unitary $\mathbf{U}_{A_0A_1}=\mathbf{U}_{A_0}\otimes \mathbf{U}_{A_1}$, the interaction class $\mathscr{I}_{A_0:A_1}(\mathbf{U}_{A_0}\otimes \mathbf{U}_{A_1})$ includes the trivial unitary, which acts on an empty input system and produces an empty output system.

We will be particularly interested in understanding the complexity of what happens inside of the scattering region. 
To capture this, we will define a notion of complexity which is minimized over the interaction class. 
\begin{definition}\label{def:entangledpartcomplexity}
Given any notion of complexity for a unitary $\mathbf{U}$, call it the type $X$ complexity $\mathscr{C}(\mathbf{U})$, we define
\begin{align}
    \mathscr{C}_{A_0:A_1}(\mathbf{U}) \defi \min_{\mathbf{U}^I\in \mathscr{I}_{A_0:A_1}(\mathbf{U})}  \mathscr{C}(\mathbf{U}^I)
\end{align}
as the interaction-class type $X$ complexity of $\mathbf{U}$. We sometimes drop the argument and write $\mathscr{C}_{A_0:A_1}=\mathscr{C}_{A_0:A_1}(\mathbf{U})$ when the relevant unitary is clear from context. 
\end{definition}
As an example, we have $\mathscr{C}_{A_0:A_1}(\mathbf{U}_{A_0}\otimes \mathbf{U}_{A_1})=0$.

Unless otherwise specified, we will take the complexity $\mathscr{C}(\mathbf{U})$ to be the circuit complexity of $\mathbf{U}$.
Because we deal with finite dimensional quantum systems and only argue the complexity is related to entanglement cost up to polynomial overheads\footnote{See conjectures \ref{conjecture:weakcomplexity-entanglement} and \ref{conjecture:strongcomplexity-entanglement} for more precise statements.}, any choice of universal gate set will suffice.
Our notion of interaction-class complexity is similar to one introduced in \cite{buhrman2013garden} (and also used in \cite{cree2022code}), although the earlier notion applied only to a specific class of quantum tasks, known as $f$-routing.  
In \cref{sec:GH} we will see that our definition reduces to the one in \cite{buhrman2013garden,cree2022code} when applied to $f$-routing.

It is sometimes important to distinguish between the complexity of a unitary and the complexity of a task. 
A task is specified by the allowed inputs and a required condition on the outputs. 
In some cases, there could be many unitaries that produce the correct outputs, with possibly differing complexities. 
If we wish to discuss the complexity of a task, rather than of a unitary, we will use the minimal complexity of all possible unitaries that correctly complete the task. 
Usually, we will consider tasks where the requirement is to implement a specific unitary, so the notions of task complexity and unitary complexity coincide. 

An instructive example is to consider the case where the inputs and outputs to the task are classical.
Call the inputs $x_0$ and $x_1$, and required outputs $f_0(x_0,x_1)$ and $f_1(x_0,x_1)$.
In this case, the inputs on each side can be copied, and sent to both output locations. 
Then, the outputs $f_0(x_0,x_1)$ and $f_1(x_0,x_1)$ can be computed separately on each side, and output to Bob. 
This shows there is a method to complete the task with trivial interaction unitary, so the interaction-class complexity for this task is $0$.

\subsection{Computations in gravitating regions and non-local computation}\label{sec:QGinsideNLQC}

Recall that a scattering region $S$ is defined by the choice of state $\ket{\Psi}$ of the CFT before the decision regions, along with a choice of input and output regions. 
We focus on constraints on computation in bulk scattering regions in the setting where the initial state is the CFT vacuum $\ket{\Psi_0}$, corresponding to pure AdS geometry in the bulk. 
In this context, we have that the ridge area in the undeformed geometry is \emph{equal} to the mutual information in the boundary. 
Later, we will discuss the case where the scattering region is defined by an arbitrary semi-classical early time state. 
In that setting, we are not able to argue for the same constraint, for reasons that we point out. 

Before proceeding, we should discuss what we mean by a unitary happening within a given spacetime region. We will say $\mathbf{U}_{A_0A_1}$ is performed within $S$ if $A_0$, $A_1$ enter separately through the lower two faces of $S$ in an arbitrary state $\ket{\psi}_{A_0A_1R}$ with $R$ a distantly located reference system, and exit separately through the upper two faces in state $\mathbf{U}_{A_0A_1}\ket{\psi}_{A_0A_1R}$. Notice that this is defined with respect to a particular factorization of the Hilbert space into subsystems $A_0$ and $A_1$.

With this notion in hand, we define the following set of unitaries.
\begin{definition}
Define $QG_{A_0:A_1}(S)$ to be the set of unitaries $\mathbf{U}_{A_0A_1}$ which have a unitary in their interaction class that can be performed within scattering region $S$. 
\end{definition}
The set $QG_{A_0:A_1}(S)$ may depend not just on the geometry of the scattering region, but also on the bulk matter fields and other details of the theory it lives in. 
We are always interested in scattering regions constructed within theories with consistent boundary descriptions. 

To relate $QG_{A_0:A_1}(S)$ to non-local computation, we define another set of unitaries. 
\begin{definition}
$NLQC_{A_0:A_1}(E)$ is the set of unitaries $\mathbf{U}_{A_0A_1}$ for which there exists a resource system that can be used to implement the computation non-locally, and which has mutual information at most $E$.
\end{definition}
Note that when we say a unitary $\mathbf{U}_{A_0A_1}$ is performed non-locally, we mean with respect to the tensor product decomposition of $\mathcal{H}_A$ into $\mathcal{H}_{A_0}\otimes \mathcal{H}_{A_1}$.

The observations recalled in \cref{sec:localtononlocal} give that, for scattering regions constructed using the CFT vacuum as initial state, computations happening inside the scattering region are realized non-locally in the boundary, with the same tensor product decomposition and with a resource system whose mutual information is \emph{equal} to the area of the ridge. 
We can summarize this as
\begin{align}\label{eq:QGsubsetNLQCvac}
    QG_{A_0:A_1}(S_0) \subseteq NLQC_{A_0:A_1}(\text{area}(S_0))
\end{align}
where by the area of $S_0$ we mean the area of the ridge, measured in Planck units in the undeformed geometry. 

The situation for the non-vacuum case is less clear, and we do not believe the containment \ref{eq:QGsubsetNLQCvac} holds. 
In this case we still have
\begin{align}\label{eq:QGsubsetNLQCnon-vac}
    QG_{A_0:A_1}(S) \subseteq NLQC_{A_0:A_1}(I(\mathcal{V}_0:\mathcal{V}_1)/2)
\end{align}
but it is unclear how to relate the geometry of the scattering region to the mutual information.
In fact, in the non-vacuum case we can have ridge area strictly less than half the mutual information, and in fact the ridge area can be $\Theta(1)$ even when the mutual information is $\Theta(1/G_N)$. 
In general, our inability to give constraints on computation in the scattering region away from the vacuum case is related to the fact that the CFT can only be used to probe boundary anchored quantities.
In this context, that means we can probe what quantum tasks can be completed, and not directly probe what happens in the scattering region. 
Because the scattering region is defined by shooting light rays in from the boundary (or from extremal surfaces anchored to the boundary), it becomes difficult in some cases to relate its geometry to boundary defined quantities.
This difficulty disappears in the case where the initial state is the CFT vacuum.

We could try to write an analogue of \cref{eq:QGsubsetNLQCvac} for the non-vacuum case by for example considering the ridge area maximized over choices of operations on $\mathcal{V}_1$, $\mathcal{V}_2$, and claiming for example that $QG_{A_0:A_1}(S) \subseteq NLQC_{A_0:A_1}(\text{area}(r'))$ where $r'$ has maximal area over choices of operation.
While the area of $r'$ is upper bounded by the mutual information, it is still possible the mutual information can be arbitrarily large compared to this maximized area, and consequently unclear that any containment like this one should hold.  
Given these difficulties, we focus on the case where the initial state, on a time slice before the decision regions, is the vacuum. 
We emphasize however that because we allow Alice to perform arbitrary operations to the decision regions, the geometry may still undergo large changes during the computation.\footnote{We thank Jon Sorce and Geoff Penington for helpful discussions around this and the last paragraph.}

\section{Explicit linear lower bounds on entanglement and the CEB}\label{sec:linearbounds}

In the last section we argued that
\begin{align}
    QG_{A_0:A_1}(S_0) \subseteq NLQC_{A_0:A_1}(\text{area}(S_0)).
\end{align}
Because of this containment, lower bounding the entanglement required for a given computation to be larger than the area of $S$ shows it cannot be implemented within the scattering region. 

There are a few tasks where lower bounds on the resources needed in non-local computation have been proven \cite{buhrman2014position,beigi2011simplified,tomamichel2013monogamy,bluhm2021position}. 
All known bounds that do not make assumptions about the protocol used are linear, with each different bound applying for a different choice of task. 
Let us focus on the bound following from \cite{tomamichel2013monogamy}.

The strategy for obtaining this bound is as follows. 
We consider a general task $T$. 
We suppose that there is some entangled resource state $\rho_{\mathcal{V}_0\mathcal{V}_1}$ and corresponding protocol that completes the task with probability $1$. 
Then, we consider running the same protocol, but replacing the entangled resource state with the product state $\rho_{\mathcal{V}_0}\otimes \rho_{\mathcal{V}_1}$. 
Then there is a simple argument \cite{may2021holographic} that shows\footnote{Note that the entropy is defined using the natural logarithm here, differing from the usual quantum information convention.}
\begin{align}\label{eq:logpsuc}
    \frac{1}{2}I(\mathcal{V}_0:\mathcal{V}_1)_{\rho_{\mathcal{V}_0\mathcal{V}_1}} \geq - \ln p_{\text{suc}}(\rho_{\mathcal{V}_0}\otimes \rho_{\mathcal{V}_1}).
\end{align}
Intuitively, this says that a task which can only be completed with low probability without entanglement requires lots of entanglement to complete with high probability. 

For a specific family of quantum tasks related to monogamy-of-entanglement games \cite{tomamichel2013monogamy}, it is known how to prove upper bounds on success probability. 
In particular these consist of tasks where $A_0$ holds one end of a maximally entangled state, and $A_1$ holds a choice of measurement basis. 
Bob measures his end of the maximally entangled state in the chosen basis, and obtains outcome $b$. 
Alice's task is to return $b$ at both $\mathcal{R}_0$ and $\mathcal{R}_1$. 
Using monogamy-of-entanglement games that exploit constructions of mutually unbiased bases in large dimensions \cite{ivonovic1981geometrical,wootters1989optimal,klappenecker2003constructions}, one can obtain $p_{\text{suc}}(\rho_{\mathcal{V}_0}\otimes \rho_{\mathcal{V}_1}) \leq 1/\sqrt{d_{A_0}}$.\footnote{A similar argument based on mutually unbiased bases appears in \cite{beigi2011simplified}.}
Inserted in \cref{eq:logpsuc}, this leads to
\begin{align}\label{eq:logdlowerbound}
    \frac{1}{2}I(\mathcal{V}_0:\mathcal{V}_1) \geq \alpha \ln d_{A_0},\,\,\,\,\,\,\,\,\,\text{(quantum tasks bound)}
\end{align}
where $\alpha=1/2$.\footnote{In \cite{may2019quantum,may2021holographic} we gave a similar argument, but with a simpler choice of measurement basis in the monogamy game, that yields $\alpha\approx 0.15$.}
As a consequence of the above bound, we obtain that the inputs to this task cannot be taken larger than the corresponding mutual information in the boundary, divided by $2\alpha$. 

An observation is that, by randomly guessing Bob's measurement outcome, Alice can always achieve $p_{\text{suc}} = 1/d_{A_0}$. 
If there were a game which did no better than this lower bound on the success probability, that game would lead to $\alpha=1$ in the bound \eqref{eq:logdlowerbound}. 
Consequently, it is clear we can never prove a lower bound on the mutual information that is stronger than \cref{eq:logdlowerbound} with $\alpha=1$ using monogamy games.
Similarly, all of the bounds in  \cite{buhrman2014position,beigi2011simplified,bluhm2021position} are linear with $\alpha<1$.

It is interesting to ask how the bound \eqref{eq:logdlowerbound} is enforced by bulk physics.
To understand this, recall the \emph{covariant entropy bound} (CEB) \cite{bousso1999covariant,flanagan2000proof, bousso2014proof}, which can be stated as follows. 
Consider a codimension 2 spacelike surface $\gamma$, and consider a null surface $\Sigma$ which is swept out by null geodesics that start on $\gamma$, are orthogonal to $\gamma$, and end on $\sigma$ or on caustics. Then the covariant entropy bound says that if the expansion along $\Sigma$ is non-positive, then 
\begin{align}\label{eq:CEB}
    \frac{\text{area}[\gamma]}{4G_N} - \frac{\text{area}(\sigma)}{4G_N} \geq \Delta S[\Sigma]
\end{align}
where $\Delta S[\Sigma]$ is the matter entropy passing through $\Sigma$. 

By assuming the CEB, we can place a constraint on what happens within the scattering region. 
It turns out that this constraint is strong enough to enforce \cref{eq:logdlowerbound}. 
Recall that the two lower faces of the scattering region labelled $\Sigma^0,\Sigma^1$ are defined as portions of the boundary of the causal future of extremal surfaces (see \cref{eq:lowerfaces}). 
Because of this, both surfaces have non-positive expansion.
Further, $r$ is a surface inside of the lightsheet $\Sigma^0$, and as a consequence the generators of $\Sigma^0$ are also orthogonal to $r$. 
These facts mean that we can apply the CEB by taking $\gamma$ in \cref{eq:CEB} to be $r$, with $\Sigma^0$ the associated lightsheet.
We could consider including the subtracted term, taking $\sigma$ to be the edges $\Sigma^0\cap \Gamma^0$ and $\Sigma^1\cap \Gamma^1$, but this subtracted term does not seem to lead to a well defined bound on non-local computation. 
This is because each choice of protocol will deform the bulk geometry in the future of the null membrane in some way, and in particular change the area of $\Sigma^0\cap \Gamma^0$ and $\Sigma^1\cap \Gamma^1$ (but not the area of $r$). 
Fortunately, the CEB will be strong enough to enforce the constraint from non-local computation even with this subtracted term removed. 

Applying the CEB to the ridge, and dropping the subtracted area term, we learn that
\begin{align}
    \frac{\text{area}[r]}{4G_N} \geq \Delta S[\Sigma^0].
\end{align}
A similar bound holds for $\Sigma^1$.
Since the inputs could in general be in the maximally mixed state, the right hand side here is just $\ln d_{A_0}$.
Thus computation inside the scattering region is limited to having at most an area's worth of inputs from each side.

Recall that we can upper bound the area of the ridge in terms of the mutual information, as given in \cref{eq:mutualinfoandarea}. 
Using this, we have that the left hand side of the CEB is bounded above by $1/2$ the mutual information.
\begin{align}
    \frac{1}{2}I(\mathcal{V}_0:\mathcal{V}_1) \geq \ln d_{A_0} \,\,\,\,\,\,\,\,\,\text{(from CEB)}.
\end{align}
This bound is of the same form as \cref{eq:logdlowerbound}, but with $\alpha=1$. 
Since all monogamy-game based bounds have $\alpha\leq 1$, the CEB based statement is strong enough to enforce all bounds coming from monogamy games.  

An important goal is to understand if there are other quantum tasks for which stronger lower bounds in terms of the input dimension on the mutual information exist. 
For example linear bounds with $\alpha>1$, or other explicit bounds with a super-linear behaviour. 
See for example \cite{junge2021geometry} for work in this direction, where a `smoothness' assumption is made on the form of the non-local computation protocol, and new interesting bounds can be obtained.\footnote{These bounds are linear in the total input size, which is the relevant quantity for gravity, but exponential in the quantum part of the inputs.}  
For bounds linear in $\ln d_{A_0}$ with $\alpha>1$, one would have to carefully understand if the subtraction term in the CEB, interpreted appropriately, suffices to enforce this constraint, or if new physics is needed. 
For super-linear bounds, it's clear physics beyond the CEB is needed, since the upper bound on $\ln d_{A_0}$ from the CEB bound scales like $1/G_N$, while the new bound would scale differently in $G_N$. 

There is reason to expect proving super-linear lower bounds on entanglement in non-local computation to be a difficult problem. 
In the context of the $f$-routing tasks discussed in \cite{buhrman2013garden,cree2022code}, upper bounds on entanglement in terms of measures of the complexity of a classical function $f$ have been proven. 
Consequently, proving lower bounds on entanglement cost proves lower bounds on these measures of complexity. 
For example the protocol in \cite{cree2022code} shows that entanglement lower bounds imply lower bounds on span program size \cite{karchmer1993span}.
While some such lower bounds are known, they are never stronger than linear, and proving stronger bounds is a long standing and elusive goal in complexity theory.
Because of this, it is especially valuable to look for lower bounds of a different type. 
In particular bounds in terms of the complexity of the operation to be performed, rather than explicit bounds in terms of the input size, do not suffer the same challenges as explicit bounds. 
We begin exploring such bounds in the next section.

\section{What is the bulk dual of entanglement cost?}\label{sec:whatisthedual?}

The strongest existing lower bounds on entanglement cost in non-local computation are linear in the number of qubits of input. 
As discussed in the previous section, the bulk consequences of such bounds can be understood as enforced by the covariant entropy bound.
As we discuss in detail below, the best protocols for performing non-local computation use entanglement exponential in the input size.
In general then, we know little about the entanglement necessary to implement a unitary non-locally --- there is an exponential gap between our upper and lower bounds. 
Within this gap some property of the unitary $\mathbf{U}$ must control its entanglement cost, and we would like to understand what this property is.
Somewhat trivially, this property could just be the size of the inputs, and (for example) all unitaries could be implementable with linear entanglement, although this seems unlikely.\footnote{In particular, this would require an exponential improvement to the known general purpose protocols \cite{beigi2011simplified}. Evidence this isn't possible is given in \cite{junge2021geometry}.}

Whatever controls the entanglement cost, bulk physics must see this quantity and somehow enforce the associated constraint on computation. 
In particular if this quantity is too large, bulk physics must enforce that the computation cannot happen within the scattering region. 
If the controlling quantity is the input size, the CEB gives an appropriate bulk mechanism. 
If the controlling quantity is something else, bulk physics must see that quantity and enforce a new constraint. 

Lloyd has suggested \cite{lloyd2000ultimate} that the circuit complexity of a unitary controls how large a region is necessary to implement it, and has suggested a mechanism by which this could be enforced. 
Lloyd assumed that in gravity the implementation of a unitary $\mathbf{U}$ can be modelled as a sequence of elementary gates acting on $O(1)$ qubits, each of which maps the input to an orthogonal state. 
Given this assumption, the Margolus-Levitin theorem \cite{margolus1998maximum} gives that the time required to implement a single gate is related inversely to the expectation value of the energy of the system,\footnote{The $\lesssim$ here means we are dropping constant factors.}
\begin{align}\label{eq:MLtheorem}
    \delta t \gtrsim \frac{1}{\langle E \rangle},
\end{align}
where $E$ is the energy above the ground state. 
Lloyd then noted that the energy inside the subregion should be less than the energy of a black hole that encloses the system, say of area $R$, so $\langle E \rangle \lesssim R$. Calling the time extent of the region $T$, we would then have
\begin{align}
    \mathcal{C} \lesssim T /\delta t \lesssim T R.
\end{align}
Assuming $T\sim R$, we can also express this in terms of the regions area $A$
\begin{align}\label{eq:lloydscomplextyareabound}
    \mathcal{C} \lesssim A^{2/(d-1)}.
\end{align}
For scattering regions, the naturally associated area is the ridge area, which we expect is similar to the area of an enclosing sphere. 
It is interesting that a simple assumption about bulk computations --- that they happen via local gates that map to orthogonal states --- leads to a bound on complexity. 
We will later argue that the non-local computation perspective also suggests a bound on complexity. 

It is important to emphasize that Lloyd's assumption really is an assumption, and not a fact about quantum mechanics: even implementing high complexity unitaries need not involve moving through many orthogonal states. Further, this assumption or some additional input is needed for Lloyd's conclusion to hold, as it is possible within the framework of quantum mechanics to do arbitrarily complex unitaries arbitrarily quickly, while at arbitrarily low energy, as pointed out by Jordan \cite{jordan2017arbitrarily}. 
It would be interesting to understand if Lloyd's assumption can be better justified, perhaps by considering the covariant entropy bound and the finite speed of light. 

Starting from Lloyd's suggestion that complexity controls how large a region is needed to implement a computation, we will explore the possibility that circuit complexity controls the entanglement cost to implement a unitary non-locally. 
Before taking this up in the next section, we will formalize this idea as a pair of conjectures. 

\begin{conjecture}\label{conjecture:weakcomplexity-entanglement}
Given a unitary $\mathbf{U}_{A_0A_1}$, the minimal entanglement cost needed to implement $\mathbf{U}_{A_0A_1}$ non-locally, labelled $E_{A_0:A_1}$, and the interaction-class complexity of $\mathbf{U}$, labelled $\mathscr{C}_{A_0:A_1}$, are related by 
\begin{align}
    F_0(\mathscr{C}_{A_0:A_1}) \leq E_{A_0:A_1} \leq F_1(\mathscr{C}_{A_0:A_1}),
\end{align}
where $F_1(\mathscr{C})=\text{poly}(F_0(\mathscr{C}))$.
\end{conjecture}
We will refer to this statement as the \emph{weak complexity-entanglement conjecture}. Roughly, this says that the complexity is the relevant quantity controlling entanglement cost, without specifying the relationship between the two quantities.
When upper and lower bounds are related polynomially, as in this conjecture, we will say they are `nearly matching'.

We can also make a stronger conjecture that specifies this relationship.
In particular, suppose that bound \eqref{eq:MLtheorem} can be saturated. 
That is, assume it is possible to do gates on $O(1)$ qubits in time inversely related to the system's energy. 
Then we could saturate the bound \eqref{eq:lloydscomplextyareabound}, which suggests the following relationship between complexity and entanglement in non-local computation.
\begin{conjecture}\label{conjecture:strongcomplexity-entanglement}
Given a unitary $\mathbf{U}_{A_0A_1}$, the minimal entanglement cost needed to implement $\mathbf{U}_{A_0A_1}$ non-locally, $E_{A_0:A_1}$, and the interaction-class complexity of $\mathbf{U}_{A_0A_1}$, $\mathscr{C}_{A_0:A_1}$ are related polynomially, 
\begin{align}
    F_0(\mathscr{C}_{A_0:A_1}) \leq E_{A_0:A_1} \leq F_1(\mathscr{C}_{A_0:A_1})
\end{align}
where $F_0$ and $F_1$ are polynomial in the complexity $\mathscr{C}$.
\end{conjecture}
We will refer to this statement as the \emph{strong complexity-entanglement conjecture}.
In both conjectures we are agnostic about the specific choice of complexity $\mathscr{C}$ underlying the definition of the interaction-class complexity, although the (non-uniform) circuit complexity is a natural choice given Lloyd's argument.

\section{Complexity and entanglement in non-local computation}\label{sec:upperandlowerbounds}

To support conjectures \ref{conjecture:weakcomplexity-entanglement} and \ref{conjecture:strongcomplexity-entanglement}, we first show that there exist \emph{some} non-trivial upper and lower bounds in terms of complexity on the entanglement cost. 
For upper bounds, we use \cite{speelman2015instantaneous}. 
The lower bound we prove here. 
Together they give,
\begin{align}\label{eq:upperandlowerbounds}
    \Omega(\log \log \mathscr{C}_{A_0:A_1}) \leq E_{A_0:A_1} \leq O(\mathscr{C}_{A_0:A_1}2^{\mathscr{C}_{A_0:A_1}}).
\end{align}
where $\mathscr{C}_{A_0:A_1}$ is the interaction-class circuit complexity, $E_{A_0:A_1}$ is the entanglement cost.
The lower bound holds for a large but not fully general class of quantum tasks, where the input on one side must be classical.

As well, there are upper and lower bounds in terms of $n'$, the minimal number of qubits sent into the interaction unitary. 
These bounds are
\begin{align}\label{eq:sizebasedbounds}
    \Omega(\log n') \leq E_{A_0:A_1} \leq O(2^{8n'})
\end{align}
where we prove the lower bound here and review the upper bound, which was shown in \cite{beigi2011simplified}. 
Note that the lower bound holds for any task where the inputs on one side are classical, whereas the linear bounds discussed in \cref{sec:linearbounds} hold for specially constructed tasks. 

Our upper and lower bounds are unfortunately far from being related polynomially to each other or to the entanglement cost.
As well, aside from the complexity based bound, there is also the system size based bound of \cref{eq:sizebasedbounds}, which raises $n'$ as a possible quantity controlling the entanglement cost, and possibly undermining our argument. \Cref{sec:restrictedprotocols} and \cref{sec:froutingprediction} offer other evidence for the complexity based conjectures.

\subsection{Upper bound from complexity}\label{sec:speelmanbounds}

We briefly recall an upper bound due to Speelman \cite{speelman2015instantaneous}. Speelman considers the task of implementing a unitary $\mathbf{U}_{A_0A_1}$, given a circuit decomposition of $\mathbf{U}_{A_0A_1}$ into the Cliffords $+$ $T$ gate set. He gives two protocols, which achieve
\begin{align}\label{eq:Tcountbound}
    E_{A_0:A_1} &= O(n2^k),\\
    E_{A_0:A_1} &= O((68n)^d)
\end{align}
Here $n$ is the number of qubits in each of the inputs, $k$ is the number of $T$ gates appearing in the circuit decomposition, and $d$, is the number of layers of $T$ gates.

To implement $\mathbf{U}_{A_0A_1}$, we can first simplify as much as possible via local pre- and post-processing, and then use Speelman's protocol on a unitary in the interaction class of $\mathbf{U}_{A_0A_1}$. 
For the interaction unitary $U^I$, call the number of input qubits $n'$,  the number of $T$ gates $k'$, and the number of layers of $T$ gates $d'$. 
We can consider $n'2^{k'}$ or $(68n')^{d'}$ as (slightly obscure) measures of the complexity of the interaction unitary, and consider these as upper bounds on entanglement cost from these measures of complexity. 
Alternatively, we can rephrase these bounds in terms of the more traditional circuit complexity using
\begin{align}
    n',k',d'\leq \mathcal{C}_{A_0:A_1}, 
\end{align}
These bounds follow, respectively, from the fact that qubits with no gates acting on them can always be removed from the interaction unitary, because the number of $T$ gates is less than the total number of gates, and because the number of layers of $T$ gates is always less than the number of gates. 
Using the first bound in \cref{eq:Tcountbound}, these lead to the upper bound in \cref{eq:upperandlowerbounds}.
When the interaction unitary can be taken to be a low complexity unitary, this provides a good bound.\footnote{The bound in terms of $\mathcal{C}_{A_0:A_1}$ following from $E_{A_0:A_1}\leq (68n')^{d'}$ is weaker.}
Because the interaction unitary can be exponentially complex, this upper bound becomes doubly exponential. 
We give a different bound which is at worst singly exponential in the next section.

\subsection{Upper bound for arbitrary unitaries}\label{sec:PTbounds}

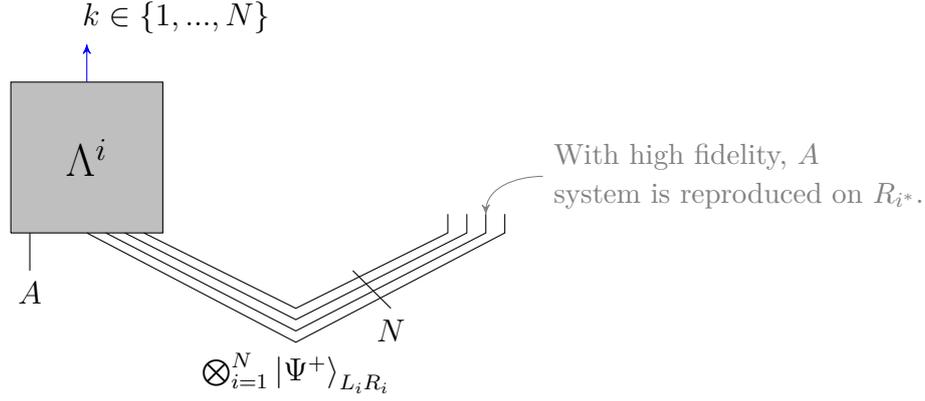
\begin{figure}
    \centering
    \begin{tikzpicture}
    
    \draw[fill=lightgray] (0,0) -- (0,2) -- (2,2) -- (2,0) -- (0,0);
    \node at (1,1) {\Large{$\Lambda^i$}};
    
    \draw (0.25,-0.5) -- (0.25,0);
    \node[below] at (0.25,-0.5) {$A$};
    
    \draw[blue,->] (1,2) -- (1,2.5);
    \node[above right] at (0.8,2.5) {$k \in \{1,...,N\}$};
    
    \draw (1.75,0) -- (3.75,-1) -- (5.75,0) -- (5.75,0.25);
    \draw (1.5,0) -- (3.75,-1.15) -- (6,0) -- (6,0.25);
    \draw (1.25,0) -- (3.75,-1.3) -- (6.25,0) -- (6.25,0.25);
    \draw (1,0) -- (3.75,-1.45) -- (6.5,0) -- (6.5,0.25);
    
    \draw (4.5,-0.5) -- (5,-1);
    \node[below] at (5,-1) {$N$};
    
    \node[below] at (3.75,-1.45) {$\bigotimes_{i=1}^N \ket{\Psi^+}_{L_iR_i}$};
    
    \draw[gray,->] (7,0.75) to [out=180,in=90] (6.25,0.25);
    \node at (7,0.75) [gray,align=left,right]{\small{With high fidelity,} $A$ \\ \small{system is reproduced on $R_{i^*}$.}};
    
    
    \end{tikzpicture}
    \caption{The port-teleportation protocol. A state $\ket{\psi}_{QA}$ is held in systems $QA$. $N$ entangled systems $\ket{\Psi^+}_{L_iR_i}$ are distributed. A POVM $\{\Lambda^i\}$ is performed on the $AL$ system producing output $i^* \in \{1,...,N\}$. The state $\ket{\psi}$ then appears on the $QR_{i^*}$ system with a fidelity controlled by $1/N$.}
    \label{fig:port_teleport}
\end{figure}

For unitaries with high complexity interaction unitarys, a different protocol performs better than Speelman's. 
This is the Beigi-K{\"o}nig protocol \cite{beigi2011simplified}, which uses as its basic tool port-teleportation \cite{ishizaka2009quantum}. 
We briefly describe port-teleportation before describing the Beigi-K{\"o}nig protocol. 
Note that our review of both topics follows \cite{may2021thesis}. 

In general, a teleportation procedure has the following form. 
\begin{enumerate*}
    \item Distribute an entangled resource state $\ket{\Psi}_{LR}$ between Alice$_0$ and Alice$_1$, where Alice$_0$ holds $L$ and Alice$_1$ holds $R$. Further, Alice$_0$ holds the $A$ subsystem of a state $\ket{\psi}_{QA}$. 
    \item Alice$_0$ performs a POVM measurement $\mathcal{M}=\{F_x\}_x$ on the $AL$ system. 
    \item Alice$_0$ sends Alice$_1$ the classical measurement outcome $x$. 
    \item Alice$_1$ applies a channel $\mathbfcal{C}^{x}_{R\rightarrow A}$.
\end{enumerate*}
The teleportation is successful when the final state on $QA$ is the same as the initial state on $QA$.

The most familiar example of a teleportation procedure occurs when $A$ is a qubit, $\ket{\Psi}_{LR}=\ket{\Psi^+}_{LR}$ is the maximally entangled state, and the measurement is in the Bell-basis, 
\begin{align}
    \mathcal{M}=\left\{\ket{\Psi^+}_{AL}, X_A \ket{\Psi^+}_{AL},Z_A\ket{\Psi^+}_{AL},X_AZ_A\ket{\Psi^+}_{AL}\right\}.
\end{align}
We will call this \emph{Bell-basis teleportation}, or just teleportation when it is clear from context that we mean this teleportation procedure specifically. An important fact about Bell-basis teleportation is that after Alice$_0$'s measurement the $QR$ system is in one of the states
\begin{align}
    Z_R^{x_1}X_R^{x_2}\ket{\psi}_{QR}.
\end{align}
Bob's correction operation is to apply $Z_R^{x_1}X_R^{x_2}$, which he can do once he receives $x=x_1x_2$ from Alice, and then relabel $R$ as $A$. 

A \emph{port-teleportation} procedure is illustrated in \cref{fig:port_teleport}. In port-teleportation, the entangled state is
\begin{align}
    \ket{\Psi}_{LR} = \bigotimes_{i=1}^N \ket{\Psi^+}_{L_iR_i}.
\end{align}
Each $L_i$, $R_i$ system has the dimensionality of $A$.
The measurement will produce an outcome $i^*\in \{1,...,N\}$, and the correction operation will be to trace out all but the $i^*$ subsystem $R_i$. 
See \cite{ishizaka2009quantum} for a description of how to choose this measurement. 
In port-teleportation $N$ may be large, so that the dimensionality of the resource system is much larger than that of the input system $A$. Further, the correction operation is the trace, which has the interesting feature that \begin{align}
    \tr_{R_{i^*}^c}([\otimes_{j=1}^N \mathbf{U}_{R_i}](\cdot)[ \otimes_{j=1}^N \mathbf{U}_{R_j}^\dagger]) = \mathbf{U}_{R_{i^*}}\tr_{R_{i^*}^c}(\cdot)\mathbf{U}^\dagger_{R_{i^*}}
\end{align} 
where the trace is over $R_{i^*}^c=R_1...R_N \setminus R_{i^*}$. We will see below that this is port-teleportation's key feature relevant for non-local computation. 

In port-teleportation, the reproduction of the $A$ system in the receiver's lab is approximate. 
Call the teleportation channel $\mathcal{T}$.
We can bound the diamond norm distance of the teleportation channel from the identity according to
\begin{align}\label{eq:PTdiamond}
    ||\mathbfcal{T}-\mathbfcal{I} ||_\diamond \leq \frac{4d_A^2}{\sqrt{N}}.
\end{align}
Notice that for a good teleportation, the \emph{number} of ports must be chosen to be large compared to the \emph{dimension} of the input system.
Consequently the number of qubits in the resource system will be exponential in the number of input qubits.  

With this background, we are ready to give the Beigi-K{\"o}nig protocol to compute a unitary $\mathbf{U}_{A_0A_1}$ non-locally.
Intuitively, the protocol is as follows. 
First, we use a Bell basis measurement to bring the $A_0$ and $A_1$ systems together on one side, say the right side, obtaining measurement outcome $x$. 
Then, we port-teleport $A_0A_1$ to the left, obtaining measurement outcome $i^*$.
On the left, we apply $\mathbf{U}(\mathbf{P}^{x}\otimes \mathcal{I})$. 
At this stage the input state, with the unitary correctly applied, is held on one port on the left, and the data about which port the state is recorded into is held on the right. 
During the communication round the $A_0$ portions of every port are sent to $\mathcal{R}_0$, and the $A_1$ portion to $\mathcal{R}_1$. 
Additionally, $i^*$ is sent to both $\mathcal{R}_0$ and $\mathcal{R}_1$. 
At the output locations, all systems but the $i^*$ port are traced out, and the subsystem from the correct port is returned. 

This protocol is described in more detail below. 

\begin{protocol}\textbf{Arbitrary unitaries using port-teleportation:}\\
Preparation phase:
\begin{enumerate*}
    \item Distribute a maximally entangled system $\ket{\Psi^+}_{F_{0}F_{1}}$ consisting of $n$ EPR pairs, with $F_0$ sent to $\mathcal{C}_0$ and $F_1$ to $\mathcal{C}_1$.
    \item Distribute a set of $N$ maximally entangled systems $\otimes_{i=1}^N \ket{\Psi^+}_{L_{i}R_{i}}$, with each of the states $\ket{\Psi^+}_{L_iR_i}$ consisting of $n$ EPR pairs, with all $L_{i}$ sent to $\mathcal{C}_0$ and all $R_{i}$ to $\mathcal{C}_1$.
\end{enumerate*}
Execution phase:
\begin{enumerate*}
    \item At $\mathcal{C}_0$, measure $A_0 F_{0}$ in the Bell basis, obtaining outcome $x$. Then Alice$_1$ holds the state
    \begin{align}
    (\mathbf{P}^x_{F_{1}}\otimes I_{A_1})\ket{\psi}_{F_{1}A_1}
    \end{align}
    at $\mathcal{C}_1$, and the index $x$ is held at $\mathcal{C}_0$. 
    \item At $\mathcal{C}_1$, perform the appropriate measurement as if port-teleporting systems $F_1A_1$ using the $N$ maximally entangled pairs $\otimes_{i=1}^N \ket{\Psi^+}_{L_iR_i}$. Call the measurement outcome $i^*$. Then $i^*$ is held at $\mathcal{C}_1$ and the $L_1...L_N$ systems of the state
    \begin{align}
        \ket{\Psi} \approx (\mathbf{P}^x\otimes I)\ket{\psi}_{L_{i^*}} \otimes  \ket{\rho}_{L_{i^*}^cX}
    \end{align} 
    are held at $\mathcal{C}_0$, where $X$ is some purifying system.
    \item At $\mathcal{C}_0$, apply $\mathbf{U} (P^x\otimes I)$ to every subsystem $L_k$. Then Alice holds
    \begin{align}
        \ket{\Psi} \approx \mathbf{U} \ket{\psi}_{L_{i^*}} \otimes  \mathbf{U}^{\otimes (N-1)}\ket{\rho}_{{L_{i^*}^c}X}
    \end{align}
    with all $L$ systems at $\mathcal{C}_0$, and $i^*$ at $\mathcal{C}_1$.
    \item Relabel the $L_k$ qubits as $A_{0,k}A_{1,k}$, and send all of the $A_{0,k}$ systems to $\mathcal{R}_0$ and all of the $A_{1,k}$ systems to $\mathcal{R}_1$. Send $i^*$ from $\mathcal{C}_1$ to both $\mathcal{R}_0$ and $\mathcal{R}_1$. 
    \item At $\mathcal{R}_0$ trace out all but the $A_{0,i^*}$ system, relabel it as $A_0$, and return it to Bob. Similarly at $\mathcal{R}_1$ trace out all but the $A_{1,i^*}$ system, relabel it as $A_1$, and return it to Bob.
\end{enumerate*}
\end{protocol}
This completes the arbitrary unitary non-locally, although the use of port-teleportation means this performs the intended unitary only approximately. 
Note that we can always apply the Beigi-K{\"o}nig protocol only to a unitary in the interaction class of the target unitary $\mathbf{U}$.
We let the minimal number of input qubits to the interaction unitary be $n'$.

Using the bound \eqref{eq:PTdiamond}, we find
\begin{align}
    ||\mathbf{U} \cdot \mathbf{U}^\dagger - \mathbfcal{N}_\mathbf{U}||_\diamond \leq \frac{2^{4n'+2}}{\sqrt{N}}
\end{align}
where $\mathbf{U} \cdot \mathbf{U}^\dagger$ is the intended (unitary) channel, and $\mathbfcal{N}_\mathbf{U}$ is the applied channel. 
We are interested in fixing the closeness with which the channel is performed, and understanding how the entanglement required scales with the number of input qubits.
Thus we fix $\epsilon\defi \frac{2^{4n'+2}}{\sqrt{N}}$, and find that $N=2^{8n'+4}/\epsilon^2$. 
Note that there is also an additional, sub-leading, linear entanglement cost of at most $n'$ EPR pairs from the first, Bell basis, teleportation. 
Thus in total we need at most $n'+2^{8n'+4}/\epsilon^2$ EPR pairs, which gives the second upper bound in \cref{eq:upperandlowerbounds}. 

\subsection{Lower bound from complexity}\label{sec:Clowerbound}

In this section we prove a lower bound on entanglement cost from the complexity of the interaction unitary, for a broad (but not fully general) class of quantum tasks. 

The proof idea is to suppose there exists a non-local implementation of a given task using $E$ EPR pairs. 
Then, we manipulate the non-local computation circuit to remove the need for EPR pairs and replace them with an interaction unitary in the new protocol. 
We describe this as a `surgery' that involves cutting the EPR pairs and then stitching the protocol back together using an interaction unitary. 
Importantly, we do the surgery in such a way that the complexity of the interaction unitary is upper bounded by a function of $E$, call it $F(E)$. 
Then the minimal complexity interaction unitary must have lower complexity than this, so $\mathscr{C}_{A_0:A_1}\leq F(E)$. 
Solving for $E$ we get a lower bound of the desired form. 

To perform the surgery, we will make use of the port-teleportation protocol described in the last section. 
As well, we will need to restrict the set of tasks that we consider to `one-sided' tasks, which are as follows.  
The input at $\mathcal{C}_0$ is both quantum systems $A_0$, $A_1$, and the input on the right is a classical string $x$, which specifies a unitary $\mathbf{U}^x_{A_0A_1}$. 
The task is to bring systems $A_0$ and $A_1$ to $\mathcal{R}_0$ and $\mathcal{R}_1$ respectively, after performing $\mathbf{U}_{A_0A_1}^x$.

Note that if we begin with the standard unitary task, with a fixed unitary $\mathbf{U}_{A_0A_1}$ and $A_1$ starting on the right, a set of $n=\log \dim A_0$ EPR pairs can be used to bring the task into the one-sided form. 
Thus if one can prove bounds relating complexity and entanglement in the one-sided task these can be translated, up to linear terms, into bounds on the standard task. 
Unfortunately because the bound we prove is so weak the linear factors would make our lower bound negative, and so trivial, so we are restricted to bounding the one-sided case.

\begin{figure*}
    \centering
    \begin{subfigure}{0.45\textwidth}
    \centering
    \begin{tikzpicture}[scale=0.5]
    
    \draw[thick] (-5,-5) -- (-5,-3) -- (-3,-3) -- (-3,-5) -- (-5,-5);
    \node at (-4,-4) {$\mathbf{V}^L$};
    
    \draw[thick] (5,-5) -- (5,-3) -- (3,-3) -- (3,-5) -- (5,-5);
    \node at (4,-4) {$\mathbf{V}^R$};
    
    \draw[thick] (5,5) -- (5,3) -- (3,3) -- (3,5) -- (5,5);
    \node at (4,4) {$\mathbf{W}^R$};
    
    \draw[thick] (-5,5) -- (-5,3) -- (-3,3) -- (-3,5) -- (-5,5);
    \node at (-4,4) {$\mathbf{W}^L$};
    
    \draw[thick] (-4.5,-3) -- (-4.5,3);
    
    \draw[thick] (4.5,-3) -- (4.5,3);
    
    \draw[thick] (-3.5,-3) to [out=90,in=-90] (3.5,3);
    
    \draw[thick] (3.5,-3) to [out=90,in=-90] (-3.5,3);
    
    \draw[thick] (-3.5,-5) to [out=-90,in=-90] (3.5,-5);
    \draw[black] plot [mark=*, mark size=3] coordinates{(0,-7.05)};
    
    \node at (-2,-6.7) {$/$};
    \node[below] at (-2.2,-7) {$E$};
    
    \draw[thick] (-4.5,-6) -- (-4.5,-5);
    \draw[thick] (4.5,-6) -- (4.5,-5);
    
    \draw[thick] (4.5,5) -- (4.5,6);
    \draw[thick] (-4.5,5) -- (-4.5,6);
    
    \draw[thick] (3.5,5) -- (3.5,6);
    \draw[thick] (-3.5,5) -- (-3.5,6);
    
    \end{tikzpicture}
    \caption{}
    \label{fig:non-localprePTsurgery}
    \end{subfigure}
    \hfill
    \begin{subfigure}{0.45\textwidth}
    \centering
    \begin{tikzpicture}[scale=0.5]
    
    \draw[thick] (-5,-5) -- (-5,-3) -- (-3,-3) -- (-3,-5) -- (-5,-5);
    \node at (-4,-4) {$\mathbf{V}^L$};
    
    \draw[thick] (6,-5) -- (6,-3) -- (3,-3) -- (3,-5) -- (6,-5);
    \node at (4.5,-4) {\small{$(V^{R,x})^{\otimes N}$}};
    
    \draw[thick] (5,8) -- (5,6) -- (3,6) -- (3,8) -- (5,8);
    \node at (4,7) {$\mathbf{W}^R$};
    
    \draw[thick] (6.5,5) -- (6.5,3) -- (4,3) -- (4,5) -- (6.5,5);
    \node at (5.25,4) {$\text{tr}_{(Y_1^{i^*})^c}$};
    
    \draw[thick] (-5,8) -- (-5,6) -- (-3,6) -- (-3,8) -- (-5,8);
    \node at (-4,7) {$\mathbf{W}^L$};
    
    \draw[thick] (-4,5) -- (-4,3) -- (-1.5,3) -- (-1.5,5) -- (-4,5);
    \node at (-2.75,4) {$\tr_{(Y_0^{i^*})^c}$};
    
    \draw[thick] (-4.5,-3) -- (-4.5,5);
    
    \draw[thick] (4.5,-3) -- (4.5,3);
    \draw[thick] (4.65,-3) -- (4.65,3);
    
    \draw[thick] (-3.5,0) to [out=90,in=-90] (3.5,5);
    \draw[thick] (-3.5,0) -- (-3.5,-3);
    
    \draw[thick] (3.5,0) to [out=90,in=-90] (-3.5,3);
    \draw[thick] (3.5,0) -- (3.5,-3);
    
    \draw[thick] (3.65,0.1) to [out=90,in=-90] (-3.3,3);
    \draw[thick] (3.65,0.1) -- (3.65,-3);
    
    \draw[thick] (-4.5,-6) -- (-4.5,-5);
    \draw[thick] (4.5,-6) -- (4.5,-5);
    
    \draw[thick] (4.5,5) -- (4.5,6);
    \draw[thick] (-4.5,5) -- (-4.5,6);
    
    \draw[thick] (3.5,5) -- (3.5,6);
    \draw[thick] (-3.5,5) -- (-3.5,6);
    
    \draw[thick] (4.5,8) -- (4.5,9);
    \draw[thick] (-4.5,8) -- (-4.5,9);
    
    \draw[thick] (3.5,8) -- (3.5,9);
    \draw[thick] (-3.5,8) -- (-3.5,9);
    
    \draw[thick] (-1,-1) -- (-1,1) -- (1,1) -- (1,-1) -- (-1,-1);
    
    \draw[thick] (-3.5,-5) to [out=-90,in=-90] (-2,-5);
    \draw[thick] (-2,-5) to [out=90,in=-90] (-0.5,-1);
    
    \draw[thick] (3.5,-5) to [out=-90,in=-90] (2,-5);
    \draw[thick] (2,-5) to [out=90,in=-90] (0.5,-1);
    
    \draw[thick] (3.65,-5) to [out=-90,in=-90] (1.85,-5);
    \draw[thick] (1.85,-5) to [out=90,in=-90] (0.35,-1);
    
    \draw[thick] (0.5,1) to [out=90,in=-90] (5.5,3);
    \draw[thick] (-0.5,1) to [out=90,in=-90] (-2.5,3);
    
    \draw[thick] (-0.6,-0.5) arc (180:0:0.6);
    \draw[thick] (0,-0.5) -- (0.5,0.5);
    
    \draw[black] plot [mark=*, mark size=3] coordinates{(-2.7,-5.42)};
    \draw[black] plot [mark=*, mark size=4] coordinates{(2.7,-5.47)};
    
    \node at (-1.25,-3) {$\backslash$};
    \node[below] at (-0.9,-3.4) {$E$};
    
    \end{tikzpicture}
    \caption{}
    \label{fig:entangledpartPTsurgery}
    \end{subfigure}
    \caption{a) The starting non-local quantum computation, using $E$ EPR pairs. b) A local computation which performs the same unitary as at left. Doubled lines indicate a set of $N$ EPR pairs, where $N$ is taken to be $\Theta(2^{8E})$. The measurement in the center is the port-teleportation measurement. The measurement outcome $i^*$ is sent to both output locations and used to trace out all but the $Y_0^{i^*}$ and $Y_1^{i^*}$ systems.}
    \label{fig:generalsurgery}
\end{figure*}
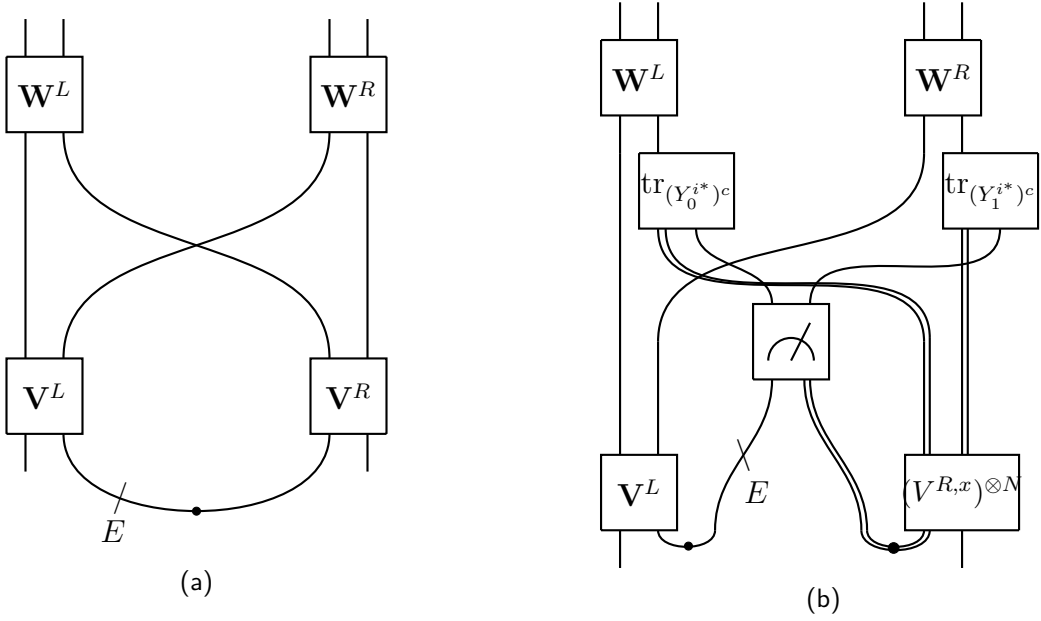

\begin{lemma}\label{lemma:PTlowerbound}
For a one-sided quantum task with the set of unitaries $\{ \mathbf{U}^x_{A_0A_1}\}_x$, the interaction-class circuit complexity of the task $\mathscr{C}_{A_0A_1:x}$ lower bounds the entanglement cost according to $\Omega(\log \log \mathscr{C}_{A_0A_1:x}) = E$.
\end{lemma}
\begin{proof}
\,Our prescription for performing surgery on the non-local computation to turn it into a local one is shown in \cref{fig:generalsurgery}. 
To describe it more carefully, we fix some notation for the non-local computation. 
In particular the non-local computation has the form
\begin{align}\label{eq:onesidednonlocal}
    \mathbf{U}_{A_0A_1}^x&\ket{\psi}_{A_0A_1}\otimes \ket{\Phi}_{Z_0Z_1} = \nonumber \\ &[\mathbf{W}^{L,x}_{X_0Y_0\rightarrow A_0Z_0}\otimes \mathbf{W}^{R,x}_{X_1Y_1\rightarrow A_1Z_1}][\mathbf{V}^L_{A_0A_1V_0\rightarrow X_0X_1}\otimes \mathbf{V}^{R,x}_{V_1\rightarrow Y_0Y_1}]\left(\ket{\psi}_{A_0A_1}\otimes \ket{\Psi^+}_{V_0V_1} \right)\nonumber 
\end{align}
where $V_0$ and $V_1$ each consist of $E$ qubits. Notice that because the input on the right is classical, we can without loss of generality consider protocols that implement an isometry $\mathbf{V}^{R,x}$ on the right that doesn't take any inputs except the $V_1$ system. 

We perform the surgery as follows. 
On the left, the state $\ket{\Psi^+}_{V_0V_1}$ consisting of $E$ EPR pairs is prepared locally.
The $V_0$ system is input into $\mathbf{V}^L$, just as in the non-local protocol. 
$V_1$ is brought into the interaction unitary. 
On the right, $N$ copies of $\ket{\Psi^+}$ are prepared locally, call them $\ket{\Psi^+}_{V_0^iV_1^i}$ with $i$ running $1$ through $N$. 
All of the $V_0^i$ systems are sent to the interaction unitary, and all of the $V_1^i$ systems are kept.
Also on the right, $\mathbf{V}^{R,x}$ is applied to each $V_1^i$, producing systems $Y_0^iY_1^i$.
Each of the $Y_0^i$ systems are sent to $\mathcal{R}_0$, each of the $Y_1^i$ systems are sent to $\mathcal{R}_1$. 
In the interaction unitary, the port-teleportation measurement is performed with $V_1$ as the input system, and the $V_0^1...V_0^N$ as the resource system. 
The measurement outcome $i^*\in\{1,...,N\}$ is sent to both output locations.
At the output locations, all but $Y_0^{i^*}$ and $Y_1^{i^*}$ systems are traced out. 
Finally, at the output locations we perform $\mathbf{W}^L_{X_0Y_0^{i^*}\rightarrow A_0Z_0}$ and $\mathbf{W}^R_{X_1Y_1^{i^*}\rightarrow A_1Z_1}$, as in the non-local computation protocol. 

To see why the modified procedure performs the same unitary as the initial non-local protocol, consider that after the port-teleportation measurement the $V_0V_1^1...V_1^N$ system is approximately in the state $\ket{\Psi^+}_{V_0V_1^{i^*}}\otimes \ket{\rho}_{(V_1^{i^*})^{c}X}$ for $X$ some purifying system. 
This is then fed into a protocol in the non-local form, which is now
\begin{align}
    [\mathbf{W}^L_{X_0Y_0^{i^*}\rightarrow A_0Z_0}\otimes \mathbf{W}^R_{X_1Y_1^{i^*}\rightarrow A_1Z_1}](\tr_{(Y_0^{i^*})^c}&\otimes \tr_{(Y_1^{i^*})^c}) \circ  \nonumber \\
    [\mathbf{V}^L_{A_0A_1V_0\rightarrow X_0X_1}\otimes &\left[\otimes_{i=1}^N(\mathbf{V}^{R,x}_{V_1^i\rightarrow Y_0^iY_1^i})\right]]\left(\ket{\Psi^+}_{V_0V_1^{i^*}} \otimes \ket{\rho}_{(V_1^{i^*})^cX} \otimes \ket{\psi}_{A_0A_1} \right). \nonumber
\end{align}
We can perform the traces to simplify this to
\begin{align}
    [\mathbf{W}^L_{X_0Y_0^{i^*}\rightarrow A_0Z_0}\otimes \mathbf{W}^R_{X_1Y_1^{i^*}\rightarrow A_1Z_1}] \circ 
    [\mathbf{V}^L_{A_0A_1V_0\rightarrow X_0X_1}\otimes \mathbf{V}^{R,x}_{V_1^{i^*}\rightarrow Y_0^iY_1^i}]\left(\ket{\Psi^+}_{V_0V_1^{i^*}} \otimes \ket{\psi}_{A_0A_1} \right), \nonumber
\end{align}
which is just the original non-local circuit of \cref{eq:onesidednonlocal}. 

To quantify how well this local protocol reproduces the unitary implemented by the non-local protocol, we again make use of the bound on the diamond norm \cref{eq:PTdiamond}. 
This leads to the channel implemented by our modified protocol, call it $\tilde{\mathcal{U}}$, agreeing with the one implemented by the non-local protocol, $\mathcal{U}$, up to
\begin{align}
    \epsilon\defi ||\mathcal{U} - \tilde{\mathcal{U}}||_\diamond \leq \frac{2^{4E+4}}{\sqrt{N}}
\end{align}
Fixing $\epsilon$ then, we need to take
\begin{align}
    N = \frac{2^{8E+8}}{\epsilon^2}
\end{align}
Given a non-local protocol using $E$ EPR pairs then, we see that there exists a local protocol where the interaction unitary consists of port-teleporting $E$ qubits using $\frac{2^{8E+8}}{\epsilon^2}$ ports. 
To upper bound complexity in terms of entanglement, we just need to upper bound the complexity of a port-teleportation in terms of the input size and number of ports, both given here in terms of $E$.

To bound the complexity of port-teleportation, we can view the measurement as a quantum channel, and then purify this channel to a unitary process on a larger system.\footnote{An alternative approach would be to note that the port-teleportation measurement is a `pretty-good measurement' and use the algorithm given in \cite{gilyen2020quantum}. Unfortunately for our specific case their bound on complexity for pretty-good measurements does no better than the naive purification argument we give here.}
Labelling $V_0V_1^1...V_1^N$ by $V$, the measurement channel is
\begin{align}
    \mathcal{M}_{V \rightarrow V X}(\rho_V)= \sum_{i} M_i \rho_V M_i^\dagger \otimes \ketbra{i}{i}_X.
\end{align}
This has $N=2^{8E+8}/\epsilon^2$ Kraus operators, $d_X=E$, and $d_V=2^{E+N}$.
This can therefore be written as a unitary acting on a system of $\log(d_Xd_VN) = O(2^{8E}/\epsilon^2)$ qubits. The maximal complexity of this unitary is $2^{O(2^{8E}/\epsilon^2)}$, and the minimal complexity of the interaction unitary must be less than the complexity used in this particular protocol, so
\begin{align}
    \mathscr{C}_{A_0A_1:x} \leq 2^{O(2^{8E}/\epsilon^2)}.
\end{align}
Inverting this we obtain 
\begin{align}
    \Omega(\log \log \mathscr{C}_{A_0:A_1}) = E,
\end{align}
as needed. 
\end{proof}

This argument also leads to a lower bound on $E$ in terms of any measure of complexity or other property of port-teleportation, so long as we minimize over choices of interaction unitary. 
In particular, we also obtain the lower bound of \cref{eq:sizebasedbounds},
\begin{align}
    \Omega(\log n')=E
\end{align}
where $n'$ is the minimal number of qubits brought into the interaction unitary. 
This follows immediately from the proof of \cref{lemma:PTlowerbound}, now noting that the port-teleportation measurement brings $O(2^{8E})$ qubits into the interaction unitary, which must be larger than $n'$, the minimal number brought into the interaction unitary. 

It is also interesting to note that while minimizing the entanglement used in port-teleportation has been carefully considered \cite{ishizaka2009quantum,christandl2021asymptotic}, to our knowledge minimizing the complexity of the needed measurement procedure has not. 
Finding a less complex measurement or better algorithm for implementing the existing measurement would lead to a better lower bound in our context. 

\section{Results for restricted classes of protocols}\label{sec:restrictedprotocols}

\Cref{conjecture:weakcomplexity-entanglement} claims that, up to polynomial overheads, the complexity of a unitary controls the entanglement cost of implementing it non-locally. 
In the last section we gave upper and lower bounds on entanglement in terms of complexity, but these are far from being polynomially related --- in fact, the upper and lower bounds are separated by three exponentials from one another.
Ideally, we would improve the upper and lower bounds until they nearly match, but so far this hasn't been possible. 

To make progress, and further explore the entanglement-complexity conjecture, in this section we consider simplifying the problem by restricting the protocols available to perform the computation non-locally.
We give two examples. 
The first considers protocols which use only Clifford unitaries. 
This restricts the protocols to perform tasks consisting of implementing a Clifford unitary. 
We recall an existing upper bound for implementing Clifford unitaries, and prove a new, nearly matching lower bound. 
Under this restriction the strong complexity-entanglement conjecture holds. 
The second example considers protocols that are limited to Bell basis measurements plus classical computations. 
This type of protocol restricts the tasks implemented to be of a type known as $f$-routing tasks. 
We recall the results of \cite{buhrman2013garden}, where they prove nearly matching upper and lower bounds for such protocols. 
This shows the weak complexity-entanglement conjecture holds in this context. 

\subsection{Strong complexity-entanglement conjecture holds for Clifford protocols}\label{sec:cliffordmatching}

In this section we show that the strong complexity-entanglement conjecture is true under a restriction to Clifford protocols. 
In particular, we show the following.
\begin{theorem}\label{thm:cliffordent-complex}
The entanglement cost to implement a Clifford unitary $\mathbf{C}_{A_0A_1}$ non-locally using a protocol consisting only of Clifford operations, $E_c$, is related to the interaction-class circuit complexity, $\mathscr{C}_{A_0:A_1}$, of that Clifford according to $\mathscr{C}_{A_0:A_1}/4 \leq E_c \leq \mathscr{C}_{A_0:A_1}$.
\end{theorem}
As well, we also show the following related fact. 
\begin{theorem}\label{thm:cliffordnprime}
Given a Clifford unitary $\mathbf{C}_{A_0A_1}$, define
\begin{align}
    n'\defi \min_{C^E\in \mathcal{E}(C)} \log dim A_0'A_1'
\end{align}
and define the entanglement cost to implement this unitary non-locally using only Clifford circuits to be $E_{c}$. Then $E_c=n'/2$. 
\end{theorem}

To be as general as possible, we discuss Clifford unitaries defined with respect to the generalized Pauli operations acting on $n$ qudits, $\mathbb{C}_d^{\otimes n}$, with $d$ prime. 
Recall that the $d$ dimensional Pauli operators $X$ and $Z$ are defined by
\begin{align}
    X\ket{j} &= \ket{j+1  \, \Mod\, d} \\
    Z\ket{j} &= \omega^j \ket{j}
\end{align}
where $\omega=e^{2\pi i/d}$. These operators generate the $d$ dimensional Pauli group. One can straightforwardly check that the generalized CNOT, Hadamard, and phase gates are in the normalizer, 
\begin{align}
    CNOT\ket{i}\ket{j} &= \ket{i}\ket{i+j \, mod \, k} \\
    H \ket{i} &= \sum_{0 \leq m\leq d-1} \omega^{m i} \ket{m} \\
    S \ket{i} &= \omega^{i(i+1)/2} \ket{i}
\end{align}
and in fact, these generate the full Clifford group over qudits \cite{clark2006valence}.

To prove theorems \ref{thm:cliffordent-complex} and \ref{thm:cliffordnprime}, we first need an upper bound on entanglement cost for implementing a Clifford non-locally. 
To perform a Clifford $\mathbf{C}_{A_0A_1}$, where each of $A_0$ and $A_1$ consist of $n$ qudits, we first reduce implementing $\mathbf{C}_{A_0A_1}$ to implementing an interaction unitary of $\mathbf{C}_{A_0A_1}$, call it $\mathbf{C}^E_{A_0'A_1'}$. Then to implement $\mathbf{C}^E_{A_0'A_1'}$ we use the following protocol \cite{chakraborty2015practical}. The input state is labelled $\ket{\psi}_{A_0'A_1'}$, where we take a pure state for notational convenience but the protocol works generally. 

\begin{protocol}\label{protocol:cliffords} \textbf{Clifford unitaries} \\

\noindent \textbf{Preparation phase:}
\begin{itemize}
    \item Prepare $n_0'=\log_d \dim A_0'$ maximally entangled states over qudits, $\ket{\Psi}_{LR}=\otimes_{i=1}^{n_0'} \ket{\Psi}_{L_iR_i}$.
\end{itemize}
\noindent \textbf{Execution phase:}
\begin{itemize}
    \item Alice$_0$ measures $A_0'L$ in the generalized Bell basis, obtaining a measurement outcome $x$. She maintains a copy of $x$ and sends a copy to Alice$_1$.
    \item Alice$_1$ relabels the $R$ system as $A_0'$, and performs the Clifford $\mathbf{C}^E_{A_0'A_1'}$. She sends the $A_0'$ system to Alice$_0$ and keeps the $A_1'$ system.
    \item After the communication round, Alice$_0$ and Alice$_1$ both locally compute $(P_y^0)_{A_0'}\otimes (P_y^1)_{A_1'} =\mathbf{C}^E_{A_0'A_1'} [(P_x)_{A_0'}\otimes \mathcal{I}_{A_1'}] (\mathbf{C}^E_{A_0'A_1'})^\dagger$
    \item Alice$_0$ applies $(P_y^0)^\dagger$ to the $A_0'$ system, then returns $A_0'$ to Bob. 
    \item Alice$_1$ applies $(P_y^1)^\dagger$ to the $A_1'$ system, then returns $A_1'$ to Bob.
\end{itemize}
\end{protocol}
To see this works correctly, notice that after Alice$_0$ measures the $A_0'L$ system in the Bell basis, and Alice$_1$ relabels $R$ as $A_0'$, Alice$_1$ holds the state $[(P_x)_{A_0'}\otimes \mathcal{I}_{A_1}] \ket{\psi}_{A_0'A_1'}$. 
She then applies $\mathbf{C}^E_{A_0'A_1'}$, giving $\mathbf{C}^E_{A_0'A_1'} [(P_x^0)_{A_0'} \otimes \mathcal{I}_{A_1'}] \ket{\psi}_{A_0'A_1'} = [(P_y)_{A_0'}\otimes (P_y)_{A_1'}] \mathbf{C}^E_{A_0'A_1'} \ket{\psi}_{A_0'A_1'}$, where we've used that $\mathbf{C}^E_{A_0'A_1'}$ is Clifford. 
After the communication round, both Alice$_0$ and Alice$_1$ can determine the Pauli string $y$, so they correct the local part of their strings and obtain $\mathbf{C}^E_{A_0'A_1'}\ket{\psi}_{A_0'A_1'}$. 
Notice that if $A_1'$ consists of fewer qudits than $A_0'$ we could have reversed their roles and teleported $A_1'$ instead.

We collect the resulting upper bound on entanglement as the following remark. 
\begin{remark}\label{remark:cliffordupper}
For a Clifford unitary $\mathbf{C}_{A_0A_1}$, protocol \ref{protocol:cliffords} gives
\begin{align}
    E_c(\mathbf{C}_{A_0A_1}) \leq \min\{n_0',n_1'\},
\end{align}
where $n'_i\defi \min_{\mathbf{C}^E\in \mathcal{E}(\mathbf{C})} \log dim A_i'$.
\end{remark}

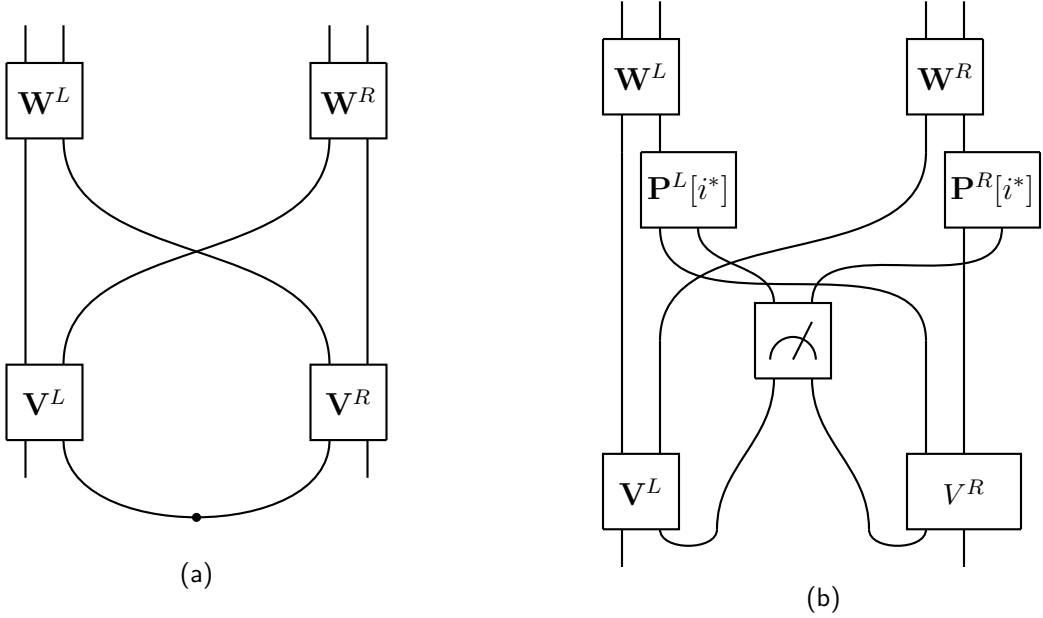
\begin{figure*}
    \centering
    \begin{subfigure}{0.45\textwidth}
    \centering
    \begin{tikzpicture}[scale=0.5]
    
    \draw[thick] (-5,-5) -- (-5,-3) -- (-3,-3) -- (-3,-5) -- (-5,-5);
    \node at (-4,-4) {$\mathbf{V}^L$};
    
    \draw[thick] (5,-5) -- (5,-3) -- (3,-3) -- (3,-5) -- (5,-5);
    \node at (4,-4) {$\mathbf{V}^R$};
    
    \draw[thick] (5,5) -- (5,3) -- (3,3) -- (3,5) -- (5,5);
    \node at (4,4) {$\mathbf{W}^R$};
    
    \draw[thick] (-5,5) -- (-5,3) -- (-3,3) -- (-3,5) -- (-5,5);
    \node at (-4,4) {$\mathbf{W}^L$};
    
    \draw[thick] (-4.5,-3) -- (-4.5,3);
    
    \draw[thick] (4.5,-3) -- (4.5,3);
    
    \draw[thick] (-3.5,-3) to [out=90,in=-90] (3.5,3);
    
    \draw[thick] (3.5,-3) to [out=90,in=-90] (-3.5,3);
    
    \draw[thick] (-3.5,-5) to [out=-90,in=-90] (3.5,-5);
    \draw[black] plot [mark=*, mark size=3] coordinates{(0,-7.05)};
    
    \draw[thick] (-4.5,-6) -- (-4.5,-5);
    \draw[thick] (4.5,-6) -- (4.5,-5);
    
    \draw[thick] (4.5,5) -- (4.5,6);
    \draw[thick] (-4.5,5) -- (-4.5,6);
    
    \draw[thick] (3.5,5) -- (3.5,6);
    \draw[thick] (-3.5,5) -- (-3.5,6);
    
    \end{tikzpicture}
    \caption{}
    \label{fig:non-localClifford}
    \end{subfigure}
    \hfill
    \begin{subfigure}{0.45\textwidth}
    \centering
    \begin{tikzpicture}[scale=0.5]
    
    \draw[thick] (-5,-5) -- (-5,-3) -- (-3,-3) -- (-3,-5) -- (-5,-5);
    \node at (-4,-4) {$\mathbf{V}^L$};
    
    \draw[thick] (6,-5) -- (6,-3) -- (3,-3) -- (3,-5) -- (6,-5);
    \node at (4.5,-4) {\small{$V^{R}$}};
    
    \draw[thick] (5,8) -- (5,6) -- (3,6) -- (3,8) -- (5,8);
    \node at (4,7) {$\mathbf{W}^R$};
    
    \draw[thick] (6.5,5) -- (6.5,3) -- (4,3) -- (4,5) -- (6.5,5);
    \node at (5.25,4) {$\mathbf{P}^R[i^*]$};
    
    \draw[thick] (-5,8) -- (-5,6) -- (-3,6) -- (-3,8) -- (-5,8);
    \node at (-4,7) {$\mathbf{W}^L$};
    
    \draw[thick] (-4,5) -- (-4,3) -- (-1.5,3) -- (-1.5,5) -- (-4,5);
    \node at (-2.75,4) {$\mathbf{P}^L[i^*]$};
    
    \draw[thick] (-4.5,-3) -- (-4.5,5);
    
    \draw[thick] (4.5,-3) -- (4.5,3);
    
    \draw[thick] (-3.5,0) to [out=90,in=-90] (3.5,5);
    \draw[thick] (-3.5,0) -- (-3.5,-3);
    
    \draw[thick] (3.5,0) to [out=90,in=-90] (-3.5,3);
    \draw[thick] (3.5,0) -- (3.5,-3);
    
    \draw[thick] (-4.5,-6) -- (-4.5,-5);
    \draw[thick] (4.5,-6) -- (4.5,-5);
    
    \draw[thick] (4.5,5) -- (4.5,6);
    \draw[thick] (-4.5,5) -- (-4.5,6);
    
    \draw[thick] (3.5,5) -- (3.5,6);
    \draw[thick] (-3.5,5) -- (-3.5,6);
    
    \draw[thick] (4.5,8) -- (4.5,9);
    \draw[thick] (-4.5,8) -- (-4.5,9);
    
    \draw[thick] (3.5,8) -- (3.5,9);
    \draw[thick] (-3.5,8) -- (-3.5,9);
    
    \draw[thick] (-1,-1) -- (-1,1) -- (1,1) -- (1,-1) -- (-1,-1);
    
    \draw[thick] (-3.5,-5) to [out=-90,in=-90] (-2,-5);
    \draw[thick] (-2,-5) to [out=90,in=-90] (-0.5,-1);
    
    \draw[thick] (3.5,-5) to [out=-90,in=-90] (2,-5);
    \draw[thick] (2,-5) to [out=90,in=-90] (0.5,-1);
    
    \draw[thick] (0.5,1) to [out=90,in=-90] (5.5,3);
    \draw[thick] (-0.5,1) to [out=90,in=-90] (-2.5,3);
    
    \draw[thick] (-0.6,-0.5) arc (180:0:0.6);
    \draw[thick] (0,-0.5) -- (0.5,0.5);
    
    \end{tikzpicture}
    \caption{}
    \label{fig:entangledpartClifford}
    \end{subfigure}
    \caption{a) The starting non-local quantum computation, using $E$ EPR pairs. Here we assume $\mathbf{V}^L, \mathbf{V}^R, \mathbf{W}^L,\mathbf{W}^R$ are all Clifford unitaries. b) A local computation which performs the same unitary as at left. The measurement in the center is a Bell basis measurement. The measurement outcome $i^*$ is sent to both output locations and used to undo the effect of the Pauli operators introduced by the Bell measurement.}
    \label{fig:Cliffordsurgery}
\end{figure*}

Next, we will give a lower bound on entanglement cost in terms of $n_0'$ and $n_1'$. This is done in the following lemma. 
\begin{lemma}\label{lemma:nprimelower}
When using a protocol consisting only of Clifford operations, the entanglement cost $E_c$ is lower bounded according to both $(n_0'+n_1')/2 \leq E_c$ and $\mathscr{C}_{A_0:A_1}/4\leq E_c$.
\end{lemma}
\begin{proof}\,
Suppose we have a non-local computation protocol which implements $\mathbf{C}_{A_0A_1}$ using $E_c$ EPR pairs, and only uses Clifford unitaries. We will show that this protocol can be manipulated to produce a local protocol with an interaction unitary that acts on at most $2E_c$ qudits. 

To do this, first consider the general form of a non-local computation protocol, shown in \cref{fig:non-localClifford}. By our assumption, $\mathbf{V}^L, \mathbf{V}^R$ are Clifford unitaries. We will manipulate this circuit diagram to produce a new protocol that replaces the $E_c$ EPR pairs with a local computation. 

The local protocol is shown in \cref{fig:entangledpartClifford}. The basic idea is to replace the shared EPR pair with EPR pairs prepared locally on each side. The local computation then consists of measuring corresponding pairs in the Bell basis. This sews the two EPR pairs together, but inserts a random Pauli. Because $\mathbf{V}^L$ and $\mathbf{V}^R$ are Clifford, this Pauli can be corrected for before applying $\mathbf{W}^L$ and $\mathbf{W}^R$. 

The local computation acts on $2E_c$ qudits, which upper bounds the minimal number of qudits which must be used in the local computation. But this is just $n_0'+n_1'$, so $n_0'+n_1'\leq 2E_c$, giving the first lower bound.

For the second lower bound, notice that the local computation consists of measuring $E_c$ EPR pairs in the Bell basis. 
Each of these measurements can be done using one Hadamard gate, one CNOT, and two single qubit measurements in the computational basis. 
The total complexity of the measurement then is $4E_c$, so $\mathscr{C}\leq 4E_c$, giving the second lower bound.
\end{proof}

We are now ready to combine the upper and lower bounds to prove theorems \ref{thm:cliffordnprime} and \ref{thm:cliffordent-complex}. 

\vspace{0.1cm}
\begin{proof}\,\textbf{(Of \cref{thm:cliffordnprime})}
Combining \cref{lemma:nprimelower} and remark \ref{remark:cliffordupper}, we have that
\begin{align}
    \frac{n_0'+n_1'}{2}\leq E_c \leq \min\{n_0',n_1'\}
\end{align}
To avoid a contradiction, these bounds require that $n_0'=n_1'$, so we learn the interesting fact that the optimal protocol always has an interaction unitary that acts on subsystems from the left and right which are of equal size. Taking $n_0'=n_1'$ we obtain $E_c=n'$ where $n'=n_0'+n_1'$. 
\end{proof}
\vspace{0.1cm}

Next we prove \cref{thm:cliffordent-complex}. 

\vspace{0.1cm}
\begin{proof}\, \textbf{(Of \cref{thm:cliffordent-complex})}
Any qubit which does not have at least one gate acting on it can be taken out of the interaction unitary, so we also have $n'\leq \mathscr{C}_{A_0A_1}$, so from \cref{thm:cliffordnprime} we have that $E_c\leq \mathscr{C}_{A_0A_1}$.
Combining this with the complexity lower bound of \cref{lemma:nprimelower} proves the theorem.
\end{proof}
\vspace{0.1cm}

It would be interesting to understand if whenever $\mathbf{C}$ is Clifford the entanglement cost is realized by a protocol that only uses Clifford operations. 
If so, we would have $n'=E_c(\mathbf{C}) = E_{A_0:A_1}(\mathbf{C})$, so this would provide an example of an explicit bound of the form of \cref{eq:logdlowerbound} discussed in \cref{sec:linearbounds} where $\alpha=1$.

We also note that our lower bound technique works for a larger class of tasks. 
In particular, suppose that in addition to the quantum systems $A_0$ and $A_1$ input on the left and right, there are additional classical strings $x_0$ and $x_1$ on the left and right, respectively. 
The task is to implement a Clifford $\mathbf{C}_{A_0A_1}^{(x_0,x_1)}$, where the choice of Clifford is a function of the classical inputs. 
Then the same surgery technique in \cref{lemma:nprimelower} provides a lower bound $\mathscr{C}_{A_0:A_1}/4 \leq E_c$. 
We have not found a good upper bound in this case, however, and we don't believe this lower bound is tight. 
As well, the same lower bound technique applies when only $\mathbf{V}^L, \mathbf{V}^R$ are Clifford but $\mathbf{W}^L,\mathbf{W}^R$ are not. 

\subsection{Weak complexity-entanglement conjecture holds for garden-hose protocols}\label{sec:GH}

In this section we discuss another class of restricted protocols. 
In this setting, earlier work \cite{buhrman2013garden} has found upper and lower bounds on entanglement as a function of the complexity which are related to each other by a polynomial, but the entanglement itself is exponential in the complexity. 
Thus in this setting only the weak complexity-entanglement conjecture is true. 

We consider tasks that have classical and quantum inputs, and restrict the protocol to consist of classical computations and Bell basis measurements, where the choice of systems to measure can depend on the classical computations. 
Because of this restriction, the tasks we can complete become limited to what are called $f$-routing tasks, defined as follows. 
At input location $\mathcal{C}_0$, Alice$_0$ receives a quantum system, call it $Q$, along with a classical string $x\in \{0,1\}^{\times n}$.
At input location $\mathcal{C}_1$, Alice$_1$ receives a classical string $y\in \{0,1\}^{\times n}$.
To complete the task, system $Q$ should be returned at $r_{f(x,y)}$, where $f_n(x,y)\in \{0,1\}$. 

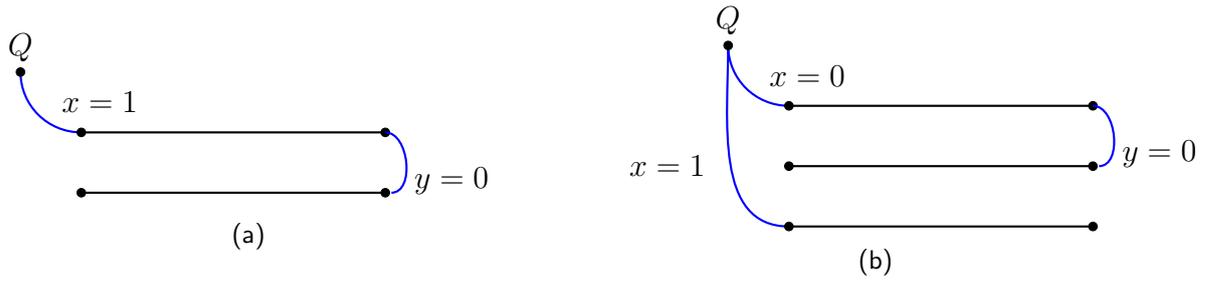
\begin{figure*}
    \centering
    \begin{subfigure}{0.45\textwidth}
    \centering
    \begin{tikzpicture}[scale=0.8]
    
    \draw[blue, thick] (-1,1)  to [out=-90,in=180] (0,0);
    
    \node[above] at (-1,1) {$Q$};
    \draw[black] plot [mark=*, mark size=2] coordinates{(-1,1)};
    
    \draw[thick] (0,0) -- (5,0);
    \draw[black] plot [mark=*, mark size=2] coordinates{(0,0)};
    \draw[black] plot [mark=*, mark size=2] coordinates{(5,0)};
    
    \draw[thick] (0,-1) -- (5,-1);
    \draw[black] plot [mark=*, mark size=2] coordinates{(0,-1)};
    \draw[black] plot [mark=*, mark size=2] coordinates{(5,-1)};

    \node[right] at (-0.5,0.5) {$x=1$}; 
    
    \draw[blue, thick] (5,0)  to [out=0,in=0] (5.1,-1);
    \node[right] at (5.3,-0.8) {$y=0$};
    
    \end{tikzpicture}
    \caption{}
    \label{fig:disconnectedsurfacesintro}
    \end{subfigure}
    \hfill
    \begin{subfigure}{0.45\textwidth}
    \centering
    \begin{tikzpicture}[scale=0.8]
    
    \draw[blue, thick] (-1,1)  to [out=-90,in=180] (0,0);
    \draw[blue, thick] (-1,1)  to [out=-90,in=180] (0,-2);
    
    \node at (-2,-1) {$x=1$};
    
    \node[above] at (-1,1) {$Q$};
    \draw[black] plot [mark=*, mark size=2] coordinates{(-1,1)};
    
    \draw[thick] (0,0) -- (5,0);
    \draw[black] plot [mark=*, mark size=2] coordinates{(0,0)};
    \draw[black] plot [mark=*, mark size=2] coordinates{(5,0)};
    
    \draw[thick] (0,-1) -- (5,-1);
    \draw[black] plot [mark=*, mark size=2] coordinates{(0,-1)};
    \draw[black] plot [mark=*, mark size=2] coordinates{(5,-1)};

    \node[right] at (-0.5,0.5) {$x=0$}; 
    
    \draw[blue, thick] (5,0)  to [out=0,in=0] (5.1,-1);
    \node[right] at (5.3,-0.8) {$y=0$};
    
    \draw[thick] (0,-2) -- (5,-2);
    \draw[black] plot [mark=*, mark size=2] coordinates{(0,-2)};
    \draw[black] plot [mark=*, mark size=2] coordinates{(5,-2)};
    
    \end{tikzpicture}
    \caption{}
    \label{fig:connectedsurfacesintro}
    \end{subfigure}
    \caption{Some simple garden-hose protocols. Blue lines indicate Bell basis measurements. Black lines indicate shared EPR pairs, with the left side of the pairs held by Alice$_0$ and right side held by Alice$_1$. a) Garden-hose protocol for computing $AND(x,y)$. Alice$_0$ measures $Q$ and the first EPR pair in the Bell basis iff $x=1$. Alice$_1$ measures the two EPR pairs iff $y=0$. b) Garden-hose protocol for $OR(x,y)$, which uses similar conditional measurements. Figure reproduced from \cite{cree2022code}.}
    \label{fig:GHexamples}
\end{figure*}

To complete a routing task locally, the simplest strategy is to bring $x$, $y$ and $Q$ together, compute the function $f$, and then direct $Q$ based on the result of the computation. 
To perform the routing task non-locally, a well known protocol is the garden-hose strategy \cite{buhrman2013garden}.
The garden-hose protocol involves sharing EPR pairs between the left and right, then doing a set of Bell measurements on 
$Q$ and entangled particles.
Which measurements are performed depends on the values of the strings $x$ and $y$.
The measurement outcomes are then communicated to both of the regions $\mathcal{R}_0$ and $\mathcal{R}_1$.
For any function $f$, if the mappings from strings $x$, $y$ to a set of measurements on both sides is chosen then $Q$ is brought to $\mathcal{R}_{f(x,y)}$.
We give simple examples of computing a NOT and AND function in \cref{fig:GHexamples}.

The entanglement cost of performing the routing task in this way has a simple relationship with the complexity of the function $f$, as was shown in \cite{buhrman2013garden}. 
To state their result we will need a different notion of complexity, which contrasts with our notion of `interaction-class' complexity from \cref{def:entangledpartcomplexity}. 
\begin{definition}\label{def:(2)complexity}
Given a classical function $f:\{0,1\}^{\times n}\times \{0,1\}^{\times n}\rightarrow \{0,1\}$ and a notion of the complexity of a classical function, call it $\mathscr{C}(\cdot)$, define the pre-processed complexity of $f$ by
\begin{align}
    \mathscr{C}_{(2)}(f) \defi \min_{\substack{M,\alpha,\beta:\\f(x,y) = M(\alpha(x),\beta(y))}}\mathscr{C}(M)\nonumber,
\end{align}
where $\alpha,\beta$ are functions mapping strings of length $n$ to strings of some length $m$.  
\end{definition}

With this notion of pre-processed complexity, we can state the result of \cite{buhrman2013garden}, which is as follows. 
Let $GH(f)$ denote the entanglement cost to perform the $f$-routing task using the garden-hose strategy.
Let SPACE$(f)$ be the memory required to compute $f$ on a Turing machine. 
Then
\begin{align}\label{eq:GHbounds}
    2^{\text{SPACE}_{(2)}(f)} \leq GH(f) \leq 2^{O(\text{SPACE}_{(2)}(f))}.
\end{align}
Notice that the upper and lower bounds in \cref{eq:GHbounds} are related polynomially, as claimed. 
Interestingly, it is not the circuit complexity but a different notion of complexity --- space needed on a Turing machine --- that appears here. 

To relate this result to our \cref{conjecture:weakcomplexity-entanglement}, we should understand the relationship between the pre-processed complexity and the interaction-class complexity $\mathscr{C}_{A_0:A_1}$. 
In particular, we need a notion of interaction-class space complexity to make a comparison. 
\begin{definition}
For an $f$-routing task, consider protocols of the local form given in \cref{fig:entangledpart} which use only classical computation and Bell basis measurements, and where the measurements may be controlled based on the bits held in memory during the classical computation. We define SPACE$_{Q,x:y}(f)$ as the minimal memory needed in the interaction unitary.
\end{definition}
With this definition, we can state the following lemma that relates interaction-class of pre-processed notions of complexity. 
\begin{lemma}\label{lemma:2vsESPACE}SPACE$_{Q,x:y}(f) = \Theta(\text{SPACE}_{(2)}(f))$.
\end{lemma}
We show this in \cref{appendix:complexitymeasures}. 

Because of \cref{lemma:2vsESPACE}, \cref{eq:GHbounds} remains true (up to some changed constants) with the pre-processed complexity of $f$ replaced by the interaction-class complexity of the $f$-routing task. 
At least under the restriction to garden-hose protocols then, \cite{buhrman2013garden} establishes polynomially related upper and lower bounds on entanglement cost in terms of interaction-class complexity. 

\section{A prediction for \texorpdfstring{$f$}{TEXT}-routing}\label{sec:froutingprediction}

\Cref{conjecture:weakcomplexity-entanglement} states that the complexity of a task controls the entanglement cost needed to implement it non-locally. 
If indeed complexity controls entanglement, it is meaningful to compare complexity-entanglement relationships among different tasks. 
In particular, in this section we point out that the Clifford protocol \ref{protocol:cliffords} efficiently implements computations of higher complexity than does the garden-hose protocol. 
If \cref{conjecture:weakcomplexity-entanglement} is true then, it would follow that improvements to the garden-hose protocol exist. 
We explore this in more detail below. 

Using polynomial entanglement, the garden-hose protocol introduced in the last section can complete the $f$-routing task for functions $f$ that, after local pre-processing, can be computed in log-space. 
We recall the following definition from \cite{buhrman2013garden} that captures the set of such functions. 
\begin{definition}
The class L$_{(2)}$ is the set of function families $\{f_n\}_n$, $f_n:\{0,1\}^{\times n}\rightarrow \{0,1\}$ such that $\spacet (f_n)=O(\log n)$, with $\spacet(\cdot)$ defined as in \cref{def:(2)complexity}. 
\end{definition}
In terms of this definition, the efficiently implementable functions in the garden-hose model are those functions in L$_{(2)}$. 

To characterize the complexity of the Clifford unitaries, consider the following computational problem.

\vspace{0.3cm}
\noindent \textbf{CliffordMeasurement$_d$:}
\begin{itemize}
    \item Input: A description of a Clifford circuit over $n$ qudits and a $Z$ eigenstate $\ket{k}$.
    \item Output: $\bra{0} \mathbf{C}^\dagger \Pi_{k,n} \mathbf{C} \ket{0}$.
\end{itemize}

Here $\Pi_{k,n}$ is the projector $\ketbra{k}{k}$ on the $n$th qudit. 
The next theorem characterizes the complexity of Clifford unitaries. 

\begin{theorem} \label{thm:Cliffordcomplete}
The CliffordMeasurement$_d$ problem is complete under log-space reductions for the class $\moddl$.
\end{theorem}
This is proven in \cite{aaronson2004improved} for $d=2$ and in \cite{de2011linearized} for prime $d$. 
To clarify the meaning of complete under log-space reductions, this means that any problem in the class $\moddl$ can be mapped to an instance of Clifford measurement using a Turing machine that has logarithmic memory, and that the Clifford measurement problem is itself in the class. 
It is important that this theorem does \emph{not} say that, given a function $f$ in $\moddl$ and input $x$, that there is a single Clifford circuit $\mathbf{C}_f$ into which we can feed the input to compute $f$. Instead, the choice of Clifford circuit depends on both $f$ and $x$, and the circuit always runs on the all zero's state. 

We can now argue as follows. First, assume the weak complexity-entanglement conjecture, which says that complexity controls entanglement in non-local computation up to polynomial overheads and second, note that the Clifford's have complexity $\moddl$ and can be implemented efficiently. Together, these imply the $f$-routing task should be efficiently implementable for all functions $f$ in $\moddl$. 
Allowing for pre-processing, we should further be able to achieve the class $\moddl_{(2)}$, defined from $\moddl$ in a way analogous to $L_{(2)}$ from $L$. 

The most obvious approach to this problem is to somehow solve the $f$-routing task by using the Clifford protocol \ref{protocol:cliffords} as a subroutine, perhaps by using the log-space reduction from functions in $\moddl$ to Clifford circuits. Because the reduction is really a map from the function plus the given input string to a Clifford circuit, it is not possible (so far as we understand) to do this.
Nonetheless our claim is that it should somehow be possible to improve on the garden-hose protocol and achieve those functions in $\moddl$. 

Recently, we studied the $f$-routing task and were able to achieve exactly the expected improvement \cite{cree2022code}. 
In particular, we showed
\begin{align}\label{eq:coderoutingupper}
    E(f) \leq O(SP_{d,(2)}(f))
\end{align}
where
\begin{align}
    SP_{d,(2)} (f) = \min_{\substack{M,\alpha,\beta:\\f(x,y) = M(\alpha(x),\beta(y))}}SP_d(M)\nonumber,
\end{align}
and $SP_d(M)$ is the minimal size of a span-program \cite{karchmer1993span} over the field $\mathbb{Z}_d$ that computes $M$. 
The complexity class of functions that can be computed with polynomial-sized span programs is $\moddl$, so that here the functions for which we can perform $f$-routing using polynomial entanglement is at least $\moddl_{(2)}$, where again the added subscript accounts for performing local pre-processing of the inputs. 
This aligns exactly with the expectation raised above, which followed from the weak complexity-entanglement conjecture. 

The protocol which achieves this improvement is inspired loosely by the holographic context, where error-correction is understood to play an important role in how bulk information is recorded into the boundary. 
Roughly, the idea of the protocol is to combine quantum error-correction with the teleportation based strategies used in the garden-hose protocol. 
Two simple examples are shown in \cref{fig:CRexamples}. 
Replacing the simple threshold code used in those examples with the general access structure secret sharing schemes constructed in \cite{smith2000quantum}, it is possible to achieve the performance given in \cref{eq:coderoutingupper}. 
This is discussed in detail in \cite{cree2022code}. 

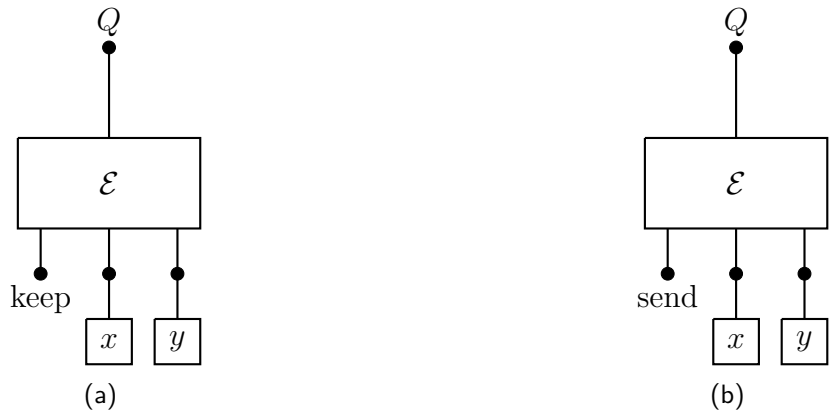
\begin{figure*}
    \centering
    \begin{subfigure}{0.45\textwidth}
    \centering
    \begin{tikzpicture}[scale=1.2]
    
    \draw[black,thick] (-3,3) -- (-1,3) -- (-1,2) -- (-3,2) -- (-3,3);
    \node at (-2,2.5) {$\mathcal{E}$};
    \draw[black,thick] (-2,4) -- (-2,3);
    \draw[black] plot [mark=*, mark size=2] coordinates{(-2,4)};
    
    \draw[black,thick] (-1.25,2) -- (-1.25,1.5);
    \draw[black,thick] (-2,2) -- (-2,1.5);
    \draw[black,thick] (-2.75,2) -- (-2.75,1.5);
    
    \draw[black] plot [mark=*, mark size=2] coordinates{(-1.25,1.5)};
    \draw[black] plot [mark=*, mark size=2]
    coordinates{(-2,1.5)};
    \draw[black] plot [mark=*, mark size=2] coordinates{(-2.75,1.5)};
    \node[below] at (-2.75,1.5) {keep};
    
    \node[above] at (-2,4) {$Q$};
    
    \draw[thick] (-1.5,1) -- (-1,1) -- (-1,0.5) -- (-1.5,0.5) -- (-1.5,1);
    \draw[thick] (-1.25,1.5) -- (-1.25,1);
    \node at (-1.25,0.75) {$y$};
    
    \draw[thick] (-2.25,1) -- (-1.75,1) -- (-1.75,0.5) -- (-2.25,0.5) -- (-2.25,1);
    \draw[thick] (-2,1.5) -- (-2,1);
    \node at (-2,0.75) {$x$};
    
    \end{tikzpicture}
    \caption{}
    \label{fig:AND}
    \end{subfigure}
    \hfill
    \begin{subfigure}{0.45\textwidth}
    \centering
    \begin{tikzpicture}[scale=1.2]
    
    \draw[black,thick] (-3,3) -- (-1,3) -- (-1,2) -- (-3,2) -- (-3,3);
    \node at (-2,2.5) {$\mathcal{E}$};
    \draw[black,thick] (-2,4) -- (-2,3);
    \draw[black] plot [mark=*, mark size=2] coordinates{(-2,4)};
    \node[above] at (-2,4) {$Q$};
    
    \draw[black,thick] (-1.25,2) -- (-1.25,1.5);
    \draw[black,thick] (-2,2) -- (-2,1.5);
    \draw[black,thick] (-2.75,2) -- (-2.75,1.5);
    
    \draw[black] plot [mark=*, mark size=2] coordinates{(-1.25,1.5)};
    \draw[black] plot [mark=*, mark size=2]
    coordinates{(-2,1.5)};
    \draw[black] plot [mark=*, mark size=2] coordinates{(-2.75,1.5)};
    
    \node[below] at (-2.75,1.5) {send};
    
    \draw[thick] (-1.5,1) -- (-1,1) -- (-1,0.5) -- (-1.5,0.5) -- (-1.5,1);
    \draw[thick] (-1.25,1.5) -- (-1.25,1);
    \node at (-1.25,0.75) {$y$};
    
    \draw[thick] (-2.25,1) -- (-1.75,1) -- (-1.75,0.5) -- (-2.25,0.5) -- (-2.25,1);
    \draw[thick] (-2,1.5) -- (-2,1);
    \node at (-2,0.75) {$x$};
    
    \end{tikzpicture}
    \caption{}
    \end{subfigure}
    \caption{Some simple code-routing protocols. The map $\mathcal{E}$ takes in the $Q$ system and records it into a 3 share secret sharing scheme where any 2 shares recover the secret. The lower boxes indicate the share should be brought to the side labelled by the corresponding bit value, which we can achieve using a garden-hose protocol that uses $O(1)$ EPR pairs. a) Code-routing protocol for computing $AND(x,y)$. b) Code-routing protocol for $OR(x,y)$. Figure reproduced from \cite{cree2022code}.}
    \label{fig:CRexamples}
\end{figure*}

Unfortunately, for code-routing protocols we were not able to prove complexity based lower bounds on entanglement cost, as is possible in for example the garden-hose case. 
Some partial results were established in \cite{cree2022code}.
It is interesting to note however that the codes constructed in \cite{smith2000quantum} are stabilizer codes, so the code-routing protocol uses only classical computations plus Clifford unitaries.
It might be possible then to prove lower bounds on code-routing using the circuit manipulation technique of \cref{sec:cliffordmatching}.
The key step would be to relate the pre-processed and interaction class span program complexity. 

\section{Discussion}\label{sec:discussion}

Broadly, our strategy in this article has been to try and understand constraints on information processing in the presence of gravity by restricting to the setting of AdS/CFT, and making use of the boundary, non-gravitational description of a given bulk information processing task.
While we have focused on possible constraints on computation within a scattering region, there are other questions we could ask about bulk information processing. 
For example, another interesting setting was discussed in \cite{may2021bulk}. 

An important question is if constraints on computation derived from non-local computation and in the context of AdS/CFT will continue to hold in more general spacetimes. 
We find this plausible, for the following reasons. 
First, the scattering region can be constructed to be much smaller than the AdS length\footnote{This is true only if the stronger claim discussed in section \ref{sec:localtononlocal}, that the standard NLQC scenario without $\mathcal{X}$ regions constrains computation, is true even away from the asymptotic setting where the $\mathcal{X}$ regions shrink.}, suggesting the bulk geometry being asymptotically AdS doesn't play a role in enforcing any constraints. 
Second, computational statements often exhibit independence from the microscopic details of the underlying systems. 
For example, many different physical systems can be used as computers, and can simulate the computations run on each other with only polynomial overheads. 
Consequently, we expect the microscopic details of the bulk theory in any specific realization of AdS/CFT to also be unimportant in enforcing these constraints. 
Given that we expect neither the detailed field content nor the asymptotic structure to play a role, we are lead to think constraints derived from non-local computation should apply in general spacetimes.

In fact, this intuition is correct for the constraints we can currently derive from non-local computation. 
In \cref{sec:linearbounds} we reviewed linear constraints on entanglement in non-local computation, which in the bulk imply no more than an area's worth of inputs can be brought into the scattering region. 
The natural bulk mechanism that enforces these constraints is the covariant entropy bound. 
In an ahistorical world where the role of non-local computation in holography was understood before the covariant entropy bound, we would, according to the reasoning of the last paragraph, be led to conjecture the CEB holds in general spacetimes. 
Indeed, all evidence thus far points to this being a correct generalization.

Returning to the question of whether or not stronger than linear bounds on entanglement cost in non-local computation exist, it is interesting to frame the problem in terms of two options, at least one of which must be true. 

Either:
\begin{enumerate}
    \item There are no constraints on computation in the bulk beyond the covariant entropy bound, or...
    \item ...there is a so far not understood constraint on computation in the bulk.
\end{enumerate}
In the boundary, there are also two mutually exclusive possibilities,
\begin{enumerate}
    \item[a.] All computations can be implemented non-locally with entanglement which is no larger than their inputs plus a description of the unitary to be performed, or...
    \item[b.] ... some unitaries require more entanglement than the size of their inputs plus the size of their description. 
\end{enumerate}
Notice that we have framed the first possibility to account for the size of a description of the unitary, which must be brought into the region and contributes towards the entropy bounded by the CEB.
Notice also that our reasoning shows that, at least for scattering regions constructed with the initial CFT state being vacuum,  $1.\Rightarrow a.$ and $b. \Rightarrow 2.$, both by applying $QG(S_0)\subseteq NLQC(S_0)$.
We find both possibilities a. and b. above exciting: the first is surprising from a quantum information theoretic perspective, and has severe consequences for the cryptographic task of position-verification (see below for comments). The second is exciting from a gravitational perspective, in that it implies 2.

To better understand these possibilities, we have tried to identify plausible properties of a unitary that could control it's entanglement cost. 
In particular, we raised two conjectures which assert a role for the complexity in determining the entanglement cost of a unitary. 
We briefly summarize the evidence for the weak complexity-entanglement conjecture (\cref{conjecture:strongcomplexity-entanglement}), which says roughly that interaction-class complexity controls entanglement cost.
\begin{itemize}
    \item From \cref{sec:upperandlowerbounds}: there exist weak upper and lower bounds on entanglement cost in terms of complexity.
    \item From \cref{sec:cliffordmatching}: restricting to protocols that use Clifford unitaries, entanglement cost is upper and lower bounded by complexity, with the upper and lower bounds being related linearly. 
    \item From \cref{sec:GH}: restricting to protocols that use Bell basis measurements, polynomially related upper and lower bounds on entanglement can be proven in terms of complexity. 
    \item From \cref{sec:froutingprediction}: as a consequence of the weak complexity-entanglement conjecture, we find that $f$-routing should be possible for all functions $f$ in $\moddl_{(2)}$ using polynomial entanglement. After this prediction was noticed, we were able to verify it \cite{cree2022code}.
\end{itemize}

The strong complexity-entanglement conjecture (\cref{conjecture:weakcomplexity-entanglement}) states that interaction-class complexity and entanglement cost should be related polynomially.
We briefly summarize the evidence. 
\begin{itemize}
    \item From \cref{sec:whatisthedual?}: there is a plausible, simple, assumption about bulk physics that would enforce the resulting upper bound on complexity in the scattering region.
    \item From the bound $E = O (2^{8n'})$: the maximal entanglement cost is polynomial in the maximal complexity (they are both exponential in $n'$, the input size). As well, whenever interaction-class complexity is zero the entanglement cost is also zero, and vice versa.
    \item From \cref{sec:cliffordmatching}: restricting to protocols that use Clifford unitaries, entanglement cost and complexity are related polynomially. 
\end{itemize}
We discuss some future directions aimed at better understanding these conjectures in the next section. 

We also note that resolving these conjectures, and better understanding entanglement cost in non-local computation generally, has important implications for position-verification \cite{kent2006tagging,kent2011quantum,malaney2010location,buhrman2014position}. 
In that context, Bob wishes to verify that Alice is able to perform quantum operations within a specified spacetime region. 
To do this, Bob arranges a quantum task of his choosing.
If Alice behaves honestly, she will complete the task locally by entering the given spacetime region and performing the needed operations.  
If Alice behaves dishonestly, she will use a non-local computation to complete the task while remaining outside of the spacetime region. 

Ideally, Bob chooses a task such that the local strategy is much easier in some way than the non-local strategy, so that it is much easier to be honest than to cheat.
An ideal task for the cryptographic setting then is one which has low complexity but high entanglement cost. 
If \cref{conjecture:strongcomplexity-entanglement} is true, this will in one sense never be the case, since entanglement cost and complexity would be related polynomially. 
Of course, one could still consider settings where, for example, the local computation is mostly classical and of polynomial complexity, while the non-local computation requires polynomial entanglement. 
Assuming classical computation is easy and storing (even polynomial) entanglement is hard then still leads to a gap between the difficulty of honest and dishonest strategies. 
For work in this direction see \cite{bluhm2021position}. 

\section{Future directions}\label{sec:futuredirections}

There are many open avenues towards better understanding the relationship between complexity, entanglement, and computation in the presence of gravity. We list a few directions below. 

\vspace{0.3cm}
\noindent \textbf{Commuting algebras vs tensor products}
\vspace{0.3cm}

While we are interested in AdS/CFT, concretely we have studied non-local quantum computation assuming that the two players Alice$_0$ and Alice$_1$ hold factorized Hilbert spaces. 
In AdS/CFT, the subsystems held by Alice$_0$ and Alice$_1$ (including their resource states) instead are described by commuting subalgebras.
For at least some information processing tasks, in particular non-local games \cite{ji2021mip}, there can be a gap between the best strategies in the tensor product and general von Neumann algebra cases.
One basic question is to understand if a similar improvement by generalizing beyond tensor product structures occurs in the context of non-local computation. Note however that the examples in \cite{ji2021mip} rely on algebras that fail to have the split property. Since the relevant regions here, $\mathcal{V}_0$ and $\mathcal{V}_1$, are separated, we naively don't expect a similar gap.\footnote{We thank Patrick Hayden for making this point to us.}

\vspace{0.3cm}
\noindent \textbf{Improved lower bounds on entanglement}
\vspace{0.3cm}

In this work we gave the first lower bound on entanglement from complexity.
Qualitatively, this is different from existing bounds, which are either for specific tasks \cite{tomamichel2013monogamy,beigi2011simplified} or specific families of tasks \cite{bluhm2021position}.
While it is interesting a lower bound of this type exists, the bound we have given is much weaker than suggested by our `strong' complexity-entanglement conjecture.  
An interesting goal is to strengthen the lower bound.
One route towards doing so would be to understand if port-teleportation can be carried out with a low complexity measurement.
This may also be a question of independent interest in quantum information.\footnote{Since this article first appeared, an efficient protocol for port-teleportation was developed \cite{timmerman2022quantum}. Combining this result with our surgery technique leads to a lower bound of $\log \mathscr{C}\leq E$ rather than our $\log \log \mathscr{C} \leq E$.} 

\vspace{0.3cm}
\noindent \textbf{Improved upper bounds on entanglement}
\vspace{0.3cm}

The upper bound on entanglement cost from complexity is also weaker than suggested by the strong complexity-entanglement conjecture. 
Recently, we improved on the most-efficient known protocol for performing the $f$-routing task by exploiting error-correcting codes \cite{cree2022code}. 
The use of error-correction was loosely inspired by its role in AdS/CFT. 
It would be interesting to understand if more detailed features of AdS/CFT could be used to inspire other non-local computation protocols, in particular protocols which solve tasks which implement general unitaries. 
One new protocol which efficiently performs computations on unitaries that do not spread operators too much was recently given in \cite{dolev2022non}, although the relationship here to complexity is unclear.

\vspace{0.3cm}
\noindent \textbf{Relaxing restrictions on protocols}
\vspace{0.3cm}

We studied protocols restricted to Clifford operations or restricted to Bell basis measurements, and found polynomial relationships between complexity and entanglement in those settings. 
Another approach to establishing the complexity-entanglement conjecture is to continue relaxing the set of allowed protocols, and attempt to establish good upper and lower bounds in the larger setting. 
For example, it might be interesting to restrict the circuit complexity of the protocol. 
This also upper bounds the complexity of the implemented unitary, so the upper bound due to Speelman \cite{speelman2015instantaneous} applies. 
We could look for a similar lower bound, perhaps using our `surgery' strategy combined with the techniques of \cite{speelman2015instantaneous}.

\vspace{0.3cm}
\noindent \textbf{Acknowledgements}
\vspace{0.3cm}

This project has grown out of many years of discussion with a large number of people, and I have surely forgotten to attribute some ideas to their original source. 
I thank Sam Cree who was involved in an early version of this project, as well as Jon Sorce, Kfir Dolev, and Patrick Hayden for careful reviews of this manuscript at various stages, as well as invaluable discussions. 
I also thank Adam Bouland, Anirudh Krishna, David P\'erez-Garcia, Aleksander Marcin Kubicki, Geoff Penington and Adam Brown for helpful discussions. 
I am supported by the Simons Foundation It from Qubit collaboration, a PDF fellowship provided by Canada's National Science and Engineering Research council, and by Q-FARM. 

\appendix

\section{Pre-processed complexity and interaction-class complexity}\label{appendix:complexitymeasures}

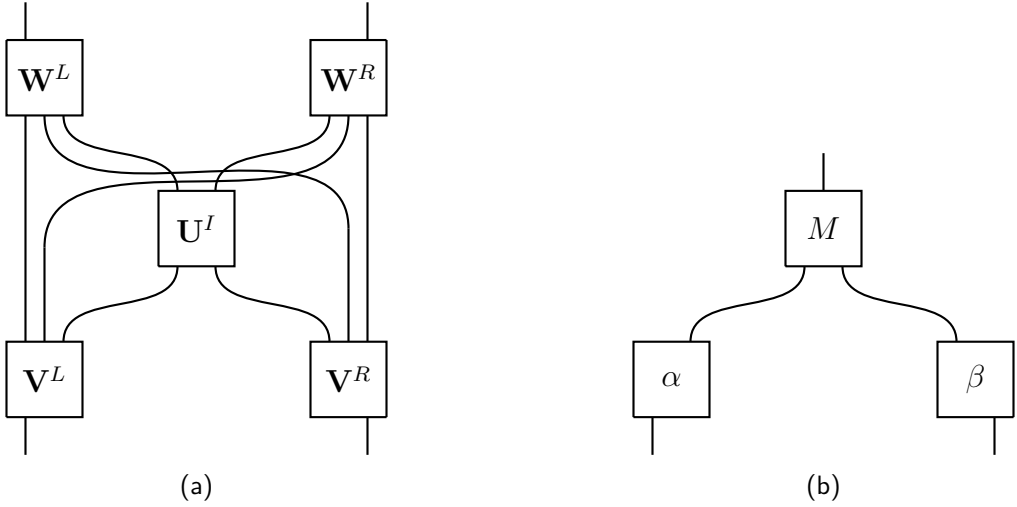
\begin{figure*}
    \centering
    \begin{subfigure}[b]{0.45\textwidth}
    \centering
    \begin{tikzpicture}[scale=0.5]
    
    \draw[thick] (-5,-5) -- (-5,-3) -- (-3,-3) -- (-3,-5) -- (-5,-5);
    \node at (-4,-4) {$\mathbf{V}^L$};
    
    \draw[thick] (5,-5) -- (5,-3) -- (3,-3) -- (3,-5) -- (5,-5);
    \node at (4,-4) {$\mathbf{V}^R$};
    
    \draw[thick] (5,5) -- (5,3) -- (3,3) -- (3,5) -- (5,5);
    \node at (4,4) {$\mathbf{W}^R$};
    
    \draw[thick] (-5,5) -- (-5,3) -- (-3,3) -- (-3,5) -- (-5,5);
    \node at (-4,4) {$\mathbf{W}^L$};
    
    \draw[thick] (-4.5,-3) -- (-4.5,3);
    
    \draw[thick] (4.5,-3) -- (4.5,3);
    
    \draw[thick] (-4,-0.5) to [out=90,in=-90] (4,3);
    \draw[thick] (-4,-0.5) -- (-4,-3);
    
    \draw[thick] (4,0) to [out=90,in=-90] (-4,3);
    \draw[thick] (4,0) -- (4,-3);
    
    \draw[thick] (-4.5,-6) -- (-4.5,-5);
    \draw[thick] (4.5,-6) -- (4.5,-5);
    
    \draw[thick] (4.5,5) -- (4.5,6);
    \draw[thick] (-4.5,5) -- (-4.5,6);
    
    
    \draw[thick] (-1,-1) -- (-1,1) -- (1,1) -- (1,-1) -- (-1,-1);
    
    \draw[thick] (-3.5,-3) to [out=90,in=-90] (-0.5,-1);
    \draw[thick] (3.5,-3) to [out=90,in=-90] (0.5,-1);
    
    \draw[thick] (0.5,1) to [out=90,in=-90] (3.5,3);
    \draw[thick] (-0.5,1) to [out=90,in=-90] (-3.5,3);
    
    \node at (0,0) {$\mathbf{U}^I$};
    
    \end{tikzpicture}
    \caption{}
    \label{fig:non-localSPACE}
    \end{subfigure}
    \hfill
    \begin{subfigure}[b]{0.45\textwidth}
    \centering
    \begin{tikzpicture}[scale=0.5]
    
    \draw[thick] (-5,-5) -- (-5,-3) -- (-3,-3) -- (-3,-5) -- (-5,-5);
    \node at (-4,-4) {$\alpha$};
    
    \draw[thick] (5,-5) -- (5,-3) -- (3,-3) -- (3,-5) -- (5,-5);
    \node at (4,-4) {$\beta$};
    
    \draw[thick] (-4.5,-6) -- (-4.5,-5);
    \draw[thick] (4.5,-6) -- (4.5,-5);
    
    \draw[thick] (-1,-1) -- (-1,1) -- (1,1) -- (1,-1) -- (-1,-1);
    
    \draw[thick] (-3.5,-3) to [out=90,in=-90] (-0.5,-1);
    \draw[thick] (3.5,-3) to [out=90,in=-90] (0.5,-1);
    
    \draw[thick] (0,1) -- (0,2);
    
    \node at (0,0) {$M$};
    
    \end{tikzpicture}
    \caption{}
    \label{fig:preprocessonly}
    \end{subfigure}
    \caption{(a) Completing the routing task with pre- and post-processing. We restrict the protocol to use Bell basis measurements and classical computation. We are interested in the space complexity SPACE$_{Q,x:y}$ of the classical computation performed in the interaction unitary. (b) Computing a function $f$ using pre-processing. In this context we are interested in the minimal complexity of the function $M$, which we label SPACE$_{(2)}(f)$.}
    \label{fig:preprocessvsentangledpart}
\end{figure*}

In this article two notions of complexity appear, which both involve performing local operations on half the inputs or outputs, in an attempt to reduce the complexity of the remaining part that interacts both inputs. 
These are the interaction-class complexity of \cref{def:entangledpartcomplexity}, and the pre-processed complexity of \cref{def:(2)complexity}.
In this section, we prove the following theorem relating interaction-class and pre-processed space complexities, in the garden-hose setting.

\vspace{0.2cm}
\noindent \textbf{\Cref{lemma:2vsESPACE}}
\emph{SPACE$_{Q,x:y}(f) = \Theta(\text{SPACE}_{(2)}(f))$.}
\vspace{0.2cm}

\begin{proof}\,
Suppose we have a protocol for completing an $f$-routing task using only Bell basis measurements and classical computation, with pre-processed space complexity $SPACE_{(2)}(f)$. 
Then, this is already of the interaction-class form of \cref{fig:entangledpart} (with $\mathbf{W}^L, \mathbf{W}^R$ trivial), so 
\begin{align}
    SPACE_{Q,x:y}(f) \leq SPACE_{(2)}(f).
\end{align}

Next we want to prove the opposite inequality, up to constant factors. 
Assume we have a local computation circuit of the form shown in \cref{fig:entangledpart}, which defines a protocol $P$.
By assumption, this protocol consists only of Bell basis measurements and classical computations, and the classical computations performed in the interaction unitary can be performed in space SPACE$_{Q,x:y}(f)$

Now consider a new protocol $P'$, which includes the same Bell measurements and classical computations as in $P$, but adds a classical memory register $R$ which tracks the location of the $Q$ system through successive Bell measurements.
In particular, the register initially stores $`Q'$, the label for the system $Q$. 
Then, we run the classical computation associated with $P$ until it specifies that $QE_i$ should be measured in the Bell basis, where $E_iE_{j}$ is in a maximally entangled state. 
Then $E_j$ is recorded into the register, and the computation is restarted. 
We then look for a measurement involving $E_jE_k$, and record $E_k$ into the register, and so on. 
Note that initially $R$ should be held on the left, since $Q$ is on the left.  

If the system holding $Q$ is sent from the left directly to either $\mathcal{R}_0$ or $\mathcal{R}_1$, we erase the label of the system in the register and just record $0$ or $1$. 
Otherwise, $Q$ will be recorded into some system $E_i$ that is sent into the interaction unitary. 
In this case send the memory register along with the inputs to the interaction unitary.
Note that the memory added to the interaction unitary need not be bigger than SPACE$_{Q,x:y}(f)+2$. 
This is because the original computation must have a memory large enough to index the input systems, so the size of $R$ is less than SPACE$_{Q,x:y}(f)$.
The additional two bits are added to record the possibility of a direct sending of $Q$ to $\mathcal{R}_0$ or $\mathcal{R}_1$.\footnote{Use the first extra bit as a flag to denote that a direct send has occurred, and the second bit to record where it was sent.} 

In the interaction unitary, continue to use the memory register to track the location of the $Q$ system. 
Then, if this system is sent towards $\mathcal{R}_0$ output $0$, and if this system is sent towards $\mathcal{R}_1$ output $1$. 
This is now a computation in the pre-processed form, and the added memory register uses at most SPACE$_{Q,x:y}(f) + 2$, so
\begin{align}
    SPACE_{(2)}(f) \leq 2 SPACE_{Q,x:y}(f) + 2
\end{align}
so we have that SPACE$_{Q,x:y}(f) = \Theta(\text{SPACE}_{(2)}(f))$, as needed. 
\end{proof}

\bibliographystyle{unsrtnat}
\bibliography{biblio}

\end{document}